\documentclass[journal]{IEEEtran}
\IEEEoverridecommandlockouts
\usepackage{amsmath,amsfonts}
\usepackage{mathtools}
\usepackage{amsthm}
\usepackage{array}
\usepackage{textcomp}
\usepackage{stfloats}
\usepackage{url}
\usepackage{microtype}
\usepackage{verbatim}
\usepackage{graphicx}
\usepackage{cite}
\usepackage[T1]{fontenc}
\usepackage[colorlinks=true, linkcolor=blue, citecolor=blue, urlcolor=blue]{hyperref}
\usepackage{enumitem}

\usepackage[linesnumbered,lined,algoruled,commentsnumbered]{algorithm2e}

\usepackage{caption}
\usepackage{subcaption}
\usepackage{stfloats}
\usepackage{lipsum} 
\usepackage[font=small,skip=0pt]{caption}
\usepackage[export]{adjustbox}
\usepackage[inkscapelatex=false]{svg}

\SetCommentSty{mycommfont}
\theoremstyle{definition}
\newtheorem{theorem}{Theorem}

\usepackage{units}
\newcommand{\nth}[1]{{#1}^{\text{th}}}

\newcommand{\abs}[1]{\left|{#1}\right|}

\newcommand{\NBS}[0]{N_{\mathrm{\scalebox{0.5} {BS} }}}

\newcommand{\RD}[0]{r_{\mathrm{\scalebox{0.5} {RD} }}}
\newcommand{\RDA}[0]{r_{\mathrm{\scalebox{0.5} {RD1} }}}
\newcommand{\RDB}[0]{r_{\mathrm{\scalebox{0.5} {RD2} }}}
\newcommand{\RDC}[0]{r_{\mathrm{\scalebox{0.5} {RD3} }}}
\newcommand{\NRF}[0]{N_{\mathrm{\scalebox{0.5} {RF}}}}
\newcommand{\rf}[0]{r_{\mathrm{\scalebox{0.5} {F}}}}
\newcommand{\BD}[0]{r_{\mathrm{\scalebox{0.5} {BD} }}}
\newcommand{\EBRD}[0]{r_{\mathrm{\scalebox{0.5} {EBRD} }}}

\newcommand{\ar}[0]{\mathbf{a}_\mathrm{\scalebox{0.5} {URA}}}
\newcommand{\bl}[0]{\mathbf{b}_\mathrm{\scalebox{0.5} {ULA}}}
\newcommand{\br}[0]{\mathbf{b}_\mathrm{\scalebox{0.5} {URA}}}

\newcommand{\GR}[0]{\mathcal{G}_\mathrm{\scalebox{0.5} {URA}}}
\newcommand{\MUE}[0]{M_{\mathrm{\scalebox{0.5} {UE}} }}

\usepackage{acronym}
\input{acronym}

\begin{document}
\title{Analyzing URA Geometry for Enhanced Near-Field Beamfocusing and Spatial Degrees of Freedom
\author{Ahmed Hussain, Asmaa Abdallah~\IEEEmembership{Senior Member, IEEE}, Abdulkadir Celik,~\IEEEmembership{Senior Member, IEEE}, \\ Emil Bj\"ornson,~\IEEEmembership{Fellow, IEEE} and Ahmed M. Eltawil,~\IEEEmembership{Senior Member, IEEE}

\thanks{Ahmed Hussain, Asmaa Abdallah, Abdulkadir Celik, and Ahmed M. Eltawil are with Computer, Electrical, and Mathematical Sciences and Engineering (CEMSE) Division, King Abdullah University of Science and Technology (KAUST), Thuwal, 23955-6900, KSA. 
Emil Bj\"ornson is with the School of Electrical Engineering and Computer Science, KTH Royal Institute of Technology, 100 44 Stockholm, Sweden. The work of E.~Bj\"ornson was supported by the Grant 2022-04222 from the Swedish Research Council.
\\{A conference version of this work is accepted in IEEE PIMRC 2025 \cite{Hussain2025ura_nearfield}}.}
}
}

\maketitle

\begin{abstract}
With the deployment of large antenna arrays at high-frequency bands, future wireless communication systems are likely to operate in the radiative near-field. Unlike far-field beam steering, near-field beams can be focused on a spatial region with a finite depth, enabling spatial multiplexing in the range dimension. Moreover, in the line-of-sight MIMO near-field, multiple spatial degrees of freedom (DoF) are accessible, akin to a scattering-rich environment. In this paper, we derive the beamdepth for a generalized uniform rectangular array (URA) and investigate how the array geometry influences near-field beamdepth and its limits. We define the effective beamfocusing Rayleigh distance (EBRD), to present a near-field boundary with respect to beamfocusing and spatial multiplexing gains for the generalized URA. Our results demonstrate that under a fixed element count constraint, the array geometry has a strong impact on beamdepth, whereas this effect diminishes under a fixed aperture length constraint. Moreover, compared to uniform square arrays, elongated configurations such as uniform linear arrays (ULAs) yield narrower beamdepth and extend the effective near-field region defined by the EBRD. Building on these insights, we design a polar codebook for compressed-sensing-based channel estimation that leverages our findings. Simulation results show that the proposed polar codebook achieves a 2 dB NMSE improvement over state-of-the-art methods. Additionally, we present an analytical expression to quantify the effective spatial DoF in the near-field, revealing that they are also constrained by the EBRD. Notably, the maximum spatial DoF is achieved with a ULA configuration, outperforming a square URA in this regard.
\end{abstract}


\section{Introduction} \label{Introduction}
\IEEEPARstart{M}{assive} \ac{MIMO} technology is a cornerstone of \ac{5G} system advancements. Building on this foundation, future wireless networks are anticipated to leverage even larger antenna arrays, such as \ac{UM}-\ac{MIMO}, while progressing toward higher frequency spectra \cite{11095387}. As carrier frequencies increase, the corresponding wavelengths decrease, enabling the deployment of larger antenna arrays within the given array size. The electromagnetic phenomenon surrounding an antenna array is typically divided into three regions: the reactive near-field, the radiative near-field, and the far-field. In this work, we refer to the radiative near-field simply as the near-field, which extends from the edge of the reactive zone out to the Rayleigh distance. Within this region, energy propagates with spherical wavefronts and exhibits distance-dependent characteristics. Beyond the Rayleigh distance—which marks the boundary between the radiative near-field and the far-field—the curvature of wavefronts diminishes, allowing them to be approximated as planar waves with a phase error of less than $\pi/8$ \cite{10934779}. 

Qualcomm recently unveiled a 4096-element Giga-\ac{MIMO} prototype operating at $\unit[13]{GHz}$, integrated within a form factor comparable to existing 256-element \ac{5G} \acp{BS} \cite{qualcomm2024mwc}. This advancement exemplifies a broader trend in future wireless systems: scaling up antenna arrays either by increasing the number of elements—enabled by higher carrier frequencies—or by expanding the physical aperture. Both approaches inherently enlarge the radiative near-field region, making it a significant portion of the communication region. Contemporary research underscores the potential of near-field propagation to enhance both single-user and multiuser communication capacity \cite{10273772}. In the near-field regime, the finite beamdepth enables spatial separation of users located within the same angular direction, thereby improving multiuser capacity \cite{10123941} for randomly located users. Simultaneously, the curvature of spherical wavefronts facilitates the transmission of multiple spatial data layers even under \ac{LoS} conditions to a single multi-antenna device, thus augmenting single-user capacity \cite{10934778}. Given the distance-dependent characteristics of near-field propagation, this work aims to address three fundamental research questions. First, what are the spatial boundaries—specifically, the range limits across different angles—within which near-field beamfocusing and multiple \ac{DoF} can be effectively achieved? Outside these limits, the \ac{LoS} channel exhibits far-field behavior, characterized by infinite beamdepth and a single spatial DoF. Second, how does the geometry of the antenna array influence beamdepth and the extent of the near-field region, under practical constraints such as a fixed aperture size or a fixed number of antenna elements? Third, how do beam pattern characteristics in the range dimension compare to those in the angular dimension? By investigating these questions, this study seeks to derive new insights into how near-field effects can be leveraged to enhance spatial multiplexing, ultimately enabling high-capacity wireless communication systems.

Furthermore, harnessing near-field gains necessitates precise channel estimation, a challenge that becomes increasingly complex in \ac{UM}-\ac{MIMO} systems employing \ac{HBF} \cite{10934779}. While \ac{HBF} significantly reduces hardware costs and complexity by limiting the number of \ac{RF} chains, it also constrains the number of observable \ac{IQ} samples. As a result, achieving reliable channel estimates demands a higher pilot overhead, highlighting the need for efficient estimation techniques along with polar codebook designs specifically tailored to near-field conditions. While polar codebooks for \acp{ULA} have been studied in \cite{10988573,9693928}, practical systems predominantly employ \acp{URA}, which accommodate a larger number of antenna elements within a constrained physical area. The polar codebooks developed for \acp{ULA} are not directly applicable to \acp{URA} due to their complex array gain functions. Moreover, the analysis and design of \ac{URA}-based polar codebooks is challenging because of the coupling among azimuth, elevation, and range parameters.

\subsection{Related Work}
Foundational work on near-field beamforming \cite{bjornson2021primer} derived beamdepth expressions for square aperture antennas, later extended to planar arrays in \cite{10443535}. Notably, the beamdepth analyses in \cite{bjornson2021primer, 10443535} are based on \ac{CAP} antennas, where the Rayleigh distance of the array is defined as the product of the number of antenna elements and the Rayleigh distance of the individual antenna element. In contrast, we define the Rayleigh distance of the array as $\RD = \frac{2D^2}{\lambda}$, where $D$ represents the aperture length of the antenna array. Furthermore, explicit closed-form beamdepth expressions in \cite{bjornson2021primer, 10443535} are limited to boresight transmissions, underscoring the need for angle-dependent beamdepth analysis for \ac{SPD} antennas.

 The capacity of a MIMO system grows linearly with the spatial \ac{DoF}, which are influenced by both the scattering environment and the array geometry \cite{pizzo2020degrees}. For large \ac{CAP} antennas, the spatial \ac{DoF} are fundamentally bounded by $\frac{2}{\lambda}$ per meter for a linear segment and $\frac{\pi}{\lambda^2}$ per square meter for planar segments \cite{8319526}. The maximum spatial \ac{DoF} is limited by the number of antenna elements and is the same in both near and far-field \cite{10934778}. Maximum spatial \ac{DoF} in the \ac{LoS} far-field channel equals one; however, spherical wavefronts in the near-field enable access to enhanced spatial \ac{DoF}. Given a constraint on a fixed number of antenna elements, we aim to determine the angle-dependent limit within the near-field where the spatial \ac{DoF} exceeds one for different array geometries.

 Recent studies have indicated that the classical Rayleigh distance often overestimates the extent of the effective near-field region \cite{hussain2024near}, prompting the development of alternative metrics to more accurately characterize near-field effects and predict when they can be utilized to achieve more well-conditioned channels. To model a wideband near-field channel, \cite{10541333} introduced the concept of \ac{ERD}, based on far-field beamforming loss in the near-field region, and this idea was expanded to wideband phased arrays with the \ac{BAND} \cite{deshpande2022wideband}. The \ac{ERD} and the \ac{BAND} cover only a small fraction of the near-field region, as defined by the Rayleigh distance. However, these metrics are based on far-field beamforming loss and therefore do not adequately capture near-field effects. The \ac{EBRD} metric was proposed in \cite{10988573} to characterize the limits of polar-domain sparsity for \acp{ULA}. Despite these advancements, an angle-dependent spatial limit for \acp{URA}—one that jointly characterizes beamfocusing and spatial \ac{DoF}—remains undetermined.
 
Near-field beamfocusing requires geometry-dependent phase compensation to account for spherical wavefront curvature. In particular, the achievable beamdepth is governed by the effective aperture length, which is determined by the diagonal (i.e., the maximum dimension) of the \ac{URA}. In this context, further research is needed to investigate array geometries that enable narrow beamdepth and extended near-field boundaries, to guide the design of future networks that enable more users to benefit from near-field propagation. This paper addresses these gaps by exploring the interplay of array geometry and near-field electromagnetic properties. Herein, we present a comprehensive analysis of the near-field beam pattern, examining key metrics such as beamdepth, limits of the beamdepth, beam resolution, spatial \ac{DoF}, and the influence of array geometry on these parameters. In contrast to \cite{10443535}, we derive angle-dependent expressions for near-field beamfocusing. Moreover, we exploit this characterization to design polar codebooks for near-field channel estimation.

\subsection{Contribution}
This work explores near-field beamfocusing for an SPD \ac{URA} with a generic width-to-height ratio. We derive the normalized array gain using Fresnel functions, based on which we formulate a beamdepth expression. Additionally, we establish angle-dependent limits for beamfocusing and characterize the near-field beam pattern in terms of resolution and sidelobe levels across both the range and angular dimensions. Furthermore, we examine the influence of spatial focus regions and array geometry on near-field beamdepth. Drawing inspiration from the far-field codebook in the angular domain—where beams spaced by the 3 dB angular beamwidth have minimum correlation while maximizing the coverage—we propose constructing a polar codebook where codewords are separated by the beamdepth in the distance dimension. We also characterize the near-field spatial \ac{DoF} using a Fresnel-based gain function and investigate how the array geometry affects these spatial \ac{DoF}. The main contributions of this paper are summarized below: 
\begin{itemize}
    \item \textbf{Beamdepth and effective beamfocusing Rayleigh distance:} We derive a closed-form beamdepth and \ac{EBRD} expression for a \ac{URA}, establishing the maximum angle-dependent range limit for finite-depth beams. Our analysis demonstrates how array geometry influences both beamdepth and \ac{EBRD}, providing critical insights for enhancing multiuser capacity in near-field communication.
    
    \item \textbf{Characterization of near-field beam pattern:} We compare and contrast near-field beam patterns across range and angle dimensions with respect to beam resolution and sidelobe levels. Furthermore, we propose a window function that attenuates sidelobes in the range dimension. 

     \item \textbf{Spatial degrees of freedom:} We derive a Fresnel-based gain function to characterize the \ac{EDoF} across different array geometries. Furthermore, we analyze the boundary beyond which the \ac{EDoF} in the near-field and far-field become equivalent.
    
    \item \textbf{Polar codebook design:} Leveraging the near-field beamforming and beam pattern analysis, we propose a novel polar codebook tailored for near-field channel estimation in a \ac{URA}. The proposed codebook enhances channel estimation performance by constructing codewords with low column coherence while improving computational efficiency through a reduced codebook size. 
\end{itemize}

\subsection{Notations and Organization of the Paper}
Matrices, vectors, and scalars are denoted by bold uppercase, bold lowercase, and lowercase letters, respectively. The expectation, transpose, and conjugate transpose operations are denoted by $\mathbb{E}[\cdot]$, $(\cdot)^\mathrm{T}$, and $(\cdot)^\mathrm{H}$ respectively. The uniform distribution between $a$ and $b$ is annotated by $\mathcal{U}(a,b)$. Furthermore, $\odot$ represents the Hadamard product.

The remainder of the paper is structured as follows. Section~\ref{Sec_System_Model} presents the near-field channel model, followed by Section~\ref{sec-3}, which derives the beamdepth and the \ac{EBRD}. Section~\ref{sec-4} analyzes the near-field beam pattern, and Section~\ref{sec-5} characterizes the spatial \ac{DoF}. The design of the polar codebook is detailed in Section~\ref{sec-6}, and conclusions are drawn in Section~\ref{sec-7}.


\section{System Model} \label{Sec_System_Model}
We consider a narrow-band communication system equipped with a \ac{UM} scale antenna array at the \ac{BS} and a single isotropic antenna \ac{UE} as depicted in Fig. \ref{fig1_system_model}. The antenna array consists of a \ac{URA} with $\NBS$ small antennas, arranged with $N_1$ elements along the y-axis and $N_2$ elements along the z-axis, respectively. The element spacing along both axes is the same and equals $d_y =d_z =d=\frac{\lambda}{2}$. Antenna elements along the y-axis are indexed as $n_1 \in [{-\overset{\sim}{N}_{1}}, \cdots,0,\cdots, {\overset{\sim}{N}_{1}-1}]$ where ${\overset{\sim}{N}_{1}} = {\frac{N_1}{2}} $. Likewise, the antenna elements along z-axis are indexed as $n_2 \in[{-\overset{\sim}{N}_{2}}, \cdots,0,\cdots, {\overset{\sim}{N}_{2}-1}]$ where ${\overset{\sim}{N}_{2}} = {\frac{N_2}{2}}$. The aperture length of the \ac{URA} can be determined from the Pythagorean theorem as $D = d \sqrt{N_1^2 + N_2^2}$. The aspect ratio of the array is defined as $\eta = \frac{N_1}{N_2}$. When $\eta > 1$, the array exhibits a wide configuration, whereas $\eta < 1$ corresponds to a tall configuration. In the extreme cases where $\eta \gg 1$ or $\eta \ll 1$, the \ac{URA} approaches a \ac{ULA}. In particular, a \ac{ULA} is obtained when its physical length equals the aperture length $D$. The \ac{UE} subtends azimuth angle $\varphi$, elevation angle $\theta$, and is positioned at a distance $r$ from the \ac{BS}. In this paper, the minimum distance in the radiative near-field is set as $2D$ such that amplitude variations across the array are negligible, and accordingly, a \ac{USW} model is employed \cite{liu2023near}. 
\begin{figure}[t]
\centering
\includegraphics[width=1\columnwidth]{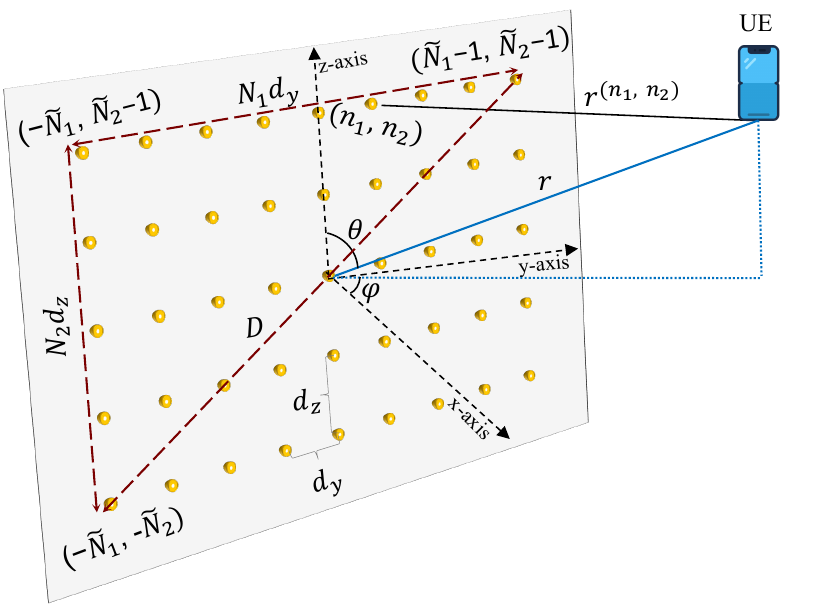}
\caption{\ac{URA} setup and a single prospective user.}
\label{fig1_system_model}
\end{figure}
The normalized near-field array response vector for a \ac{ULA} with $N_1$ antenna elements is expressed as 
\begin{equation}\small 
\bl(\varphi,r) = \tfrac{1}{\sqrt{N_1}}\Big[e^{-j\nu(r^{(-\overset{\sim}{N}_{1})} -r)},\dots, e^{-j\nu(r^{(\overset{\sim}{N}_{1})} -r)}\Big],
\label{eqn_B2}
\end{equation}
\normalsize
where $\nu = \frac{2\pi{f_c}}{c}$ is the wavenumber, with $f_c$ as the carrier frequency and $r^{(n)}$ is the distance between the \ac{UE} and the $\nth{n}$ antenna element of the \ac{ULA}. For a \ac{URA} with $\NBS = N_1 N_2 $ elements, the near-field array response vector $\br(\varphi,\theta,r)$ can be obtained based on the spherical-wave propagation model as
\begin{equation}\small
\begin{aligned}
\br(\varphi,\theta,r)
=\frac{1}{\sqrt{\NBS}}\left[e^{-j {
\nu}\left(r^{(-\tilde{\boldsymbol{\zeta}})}-r\right)}, \cdots, e^{-j {
\nu}\left(r^{(\tilde{\boldsymbol{\zeta}})}-r\right)}\right],
\label{eqn_B3}
\end{aligned}
\end{equation}
\normalsize
where $ \widetilde{\boldsymbol{\zeta}}=(\overset{\sim}{N_1},\overset{\sim}{N_2} )$. The distance $r^{(n_1,n_2)}$ between the \ac{UE} and $\nth{\left( n_1,n_2 \right)}$ element of the \ac{URA} is given by 
\begin{equation}
\begin{aligned}
r^{\left(n_{1}, n_{2}\right)} & =\sqrt{\left(r u_x-0 \right)^{2}+\left(r u_y - n_1 d\right)^{2}+\left(r u_z - n_2 d\right)^{2}}, \\
& \stackrel{(a)}{\approx} r - n_1 d u_y - n_2 d u_z + \frac{n_1^{2} d^{2}}{2r} \beta_1 \cdots \\ & + \frac{n_2^{2} d^{2}}{2r} \beta_2 - \frac{n_1 n_2 d^{2} u_y u_z}{r},\\ 
& \stackrel{(b)}{\approx} r - n_1 d u_y - n_2 d u_z + \frac{n_1^{2} d^{2}}{2r} \beta_1 + \frac{n_2^{2} d^{2}}{2r} \beta_2, 
\label{eqn_B4}
\end{aligned}
\end{equation}
where the directional cosines are $u_x = \sin \theta \cos \varphi$, $u_y = \sin \theta \sin \varphi$, and $u_z = \cos \theta$. Approximation (a), also termed the \textit{near-field expansion}, is derived from the second-order Taylor expansion $\sqrt{1+x} \approx 1+\frac{x}{2}-\frac{1}{8}x^2$, that is tight for small values of $x$. Furthermore, $\beta_1 = 1-\sin ^{2} \theta \sin ^{2} \varphi$ and, $\beta_2 = 1-\cos ^{2} \theta$. To make the following analysis tractable, we omit the last term in \eqref{eqn_B4}(a) to obtain \eqref{eqn_B4}(b), as it is negligible for large values of $\NBS$ \cite{10403506}, and accounts for only $5\%$ of the array gain \cite{10273772}. Furthermore, when $r$ is greater than the Rayleigh distance $\RD=\frac{2D^2}{\lambda}$, the terms involving $\frac{1}{r}$ can be ignored, and \eqref{eqn_B3} degenerates to the far-field steering vector
\begin{equation} \small
\begin{aligned}
 &\ar(\varphi,\theta) 
 =\frac{1}{\sqrt{\NBS}}\left[e^{-j 
\nu d\left(-\overset{\sim}{N_1} \sin \varphi \sin \theta-\overset{\sim}{N_2} \cos \theta\right)}, \cdots,1,\right. \\
&\cdots,\left.e^{-j 
\nu d\left(\overset{\sim}{N_1} \sin \varphi \sin \theta+\overset{\sim}{N_2} \cos \theta\right)}\right] \in \mathbb{C}^{N_1N_2 \times 1}.
\label{eqn_A3}
\end{aligned}
\end{equation}
\normalsize

\section{Analysis of Near-field Beamforming } \label{sec-3}
Finite-depth beamforming is achievable in the radiative near-field through the conjugate phase method. In this approach, the phase of each antenna element is adjusted to compensate for the path-length difference between the focal point and the antenna element, thereby ensuring constructive interference at the focal point. The resulting beam, as depicted in Fig. \ref{fig2_nearfield_beam}, has an azimuth beamwidth $\varphi_{\mathrm{\scalebox{0.5} {BW}}}$, elevation beamwidth $\theta_{\mathrm{\scalebox{0.5} {BW}}}$ and beamdepth $\BD$. In this section, we derive the beamdepth and the maximum range limits within which the finite depth beamforming can be achieved, referred to as the \ac{EBRD} for a \ac{URA}. We investigate how the spatial focus region and array geometry influence the variation in beamdepth and the extent of the \ac{EBRD}.

We assume a \ac{LoS} channel between the \ac{UE} and the \ac{BS}. The near-field channel $\mathbf{h} \in \mathbb{C}^{\NBS\times 1}$ based on the \ac{USW} model is mathematically formulated as \cite{sherman1962properties}
\begin{equation} 
\mathbf{h} = \sqrt{g}\mathbf{b}(\varphi,\theta,r),
\label{eqn_B1}
\end{equation}
 where $g$ denotes the channel gain taking into account the path loss and antenna gain, and $r$ is the distance between the \ac{UE} and the origin of the \ac{URA}. We consider a \ac{BS} equipped with a generalized \ac{URA} that serves a near-field \ac{UE} located at $(\varphi,\theta,\rf)$. Let $\mathbf{w}\left(\varphi,\theta,z \right) $ denote the beamforming vector based on the conjugate phase method. The array gain $\GR$ characterizes the normalized received power, and is defined as
\begin{equation}
\begin{aligned}
&\GR = \left| \mathbf{w}\left(\varphi,\theta,z \right)^\mathsf{H} \mathbf{h}\left(\varphi,\theta,\rf\right) \right|^2, \\
&\forall \ \varphi \in [-\pi/2,\ \pi/2],\ \theta \in [-\pi/2,\ \pi/2], \ z \in [2D,\ \infty]. 
\label{eqn_III}
\end{aligned}
\end{equation}
In the subsequent analysis, we utilize \eqref{eqn_III} to derive array gain for the \ac{URA} and characterize beamdepth and \ac{EBRD} for different array geometries.

\begin{figure}[t]
\centering
\includegraphics[width=1\columnwidth]{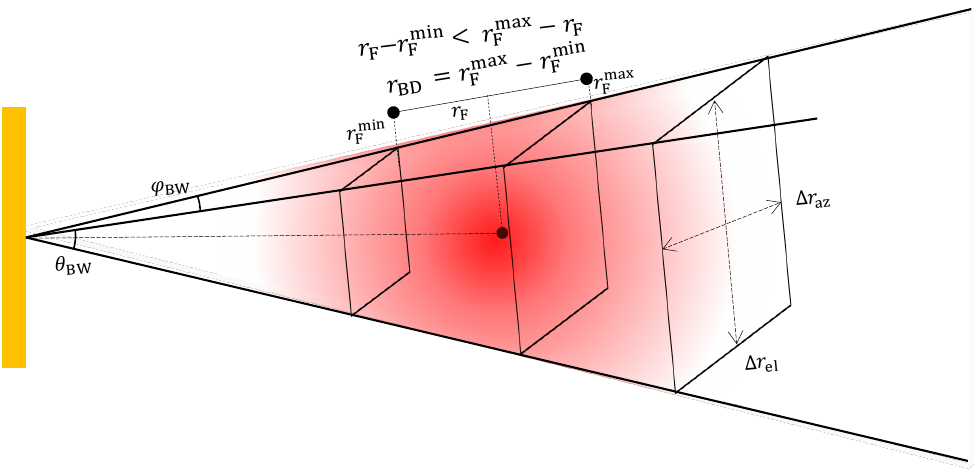}
\caption{Near-field beam with finite beamdepth and beamwidth in the axial and lateral dimensions, respectively.
}
\label{fig2_nearfield_beam}
\end{figure}

\subsection{Finite Beamdepth and EBRD} \label{sec-3a}
We define the beamdepth, $\BD$, as the distance interval $z \in [\rf^\mathrm{min},\ \rf^\mathrm{max}]$ where the normalized array gain is at most $\unit[3]{dB}$ below its maximum value. The beamforming vector $\mathbf{w}(\varphi, \theta, z)$ in (\ref{eqn_III}), designed using conjugate phase, achieves maximum array gain at the focal point $z=\rf$ and decreasing array gain at distances away from $\rf$. Then the normalized array gain in \eqref{eqn_III} achieved by $\mathbf{w}(\varphi, \theta, z)$ is given by
\begin{equation}
\GR \approx \left|\frac{1}{N_1 N_2} \sum_{n_{1}} \sum_{n_{2}} e^{j 
\frac{\pi}{\lambda}\left(n_{1}^{2} d^{2} \beta_1 + n_{2}^{2} d^{2} \beta_2 \right)
 z_\mathrm{eff} }\right|^2, 
\label{eqn_IIIB_1}
\end{equation}
where $ z_\mathrm{eff} = \left| \frac{z-\rf}{z \rf } \right| $. Since $\NBS$ is large, the summation in \eqref{eqn_IIIB_1} can be approximated with a Riemann integral\footnote{For large $\NBS$ with small antenna elements, the discrete sum converges to the Riemann integral.} to obtain
\begin{equation}\small
\begin{aligned}
\GR \approx \left|\frac{1}{N_1 N_2} \int_{-\tfrac{N_1}{2}}^{\tfrac{N_1}{2}} \int_{-\tfrac{N_2}{2}}^{\tfrac{N_2}{2}} e^{j 
\frac{\pi}{\lambda}\left(n_{1}^{2} d^{2} \beta_1 + n_{2}^{2} d^{2} \beta_2 \right)
 z_\mathrm{eff} } \mathrm{~d} n_{1} \mathrm{~d} n_{2} \right|^2.
 \label{eqn_IIIB_2}
\end{aligned}
\end{equation}
\normalsize
By the change of variables $x_1 ={\sqrt{\frac{2n_{1}^{2} d^{2} \beta_1 z_\mathrm{eff}}{\lambda} }} $ and $x_2 ={\sqrt{\frac{2n_{2}^{2} d^{2} \beta_2 z_\mathrm{eff}}{\lambda} }} $, we express the array gain in \eqref{eqn_IIIB_2} in terms of Fresnel functions as
\begin{equation}
\GR = \frac{\left[C^2(\gamma_1) + S^2(\gamma_1) \right]\left[C^2(\gamma_2) + S^2(\gamma_2)\right]}{(\gamma_1 \gamma_2)^2},
\label{eqn_IIIB_3}
\end{equation}
where $C(\gamma_i) = \int_{0}^{\gamma_i} \cos{\left(\frac{\pi}{2}x_i^2\right)}dx_i$ and $S(\gamma_i) = \int_{0}^{\gamma_i} \sin{\left(\frac{\pi}{2}x_i^2\right)}dx_i$ are the Fresnel functions \cite{10123941}, with $i \in \{1, 2\}$. Moreover, $\gamma_{1}=\sqrt{\frac{N_{1}^{2} d^{2} \beta_1 }{2\lambda}z_\mathrm{eff}}$ and $\gamma_{2}=\sqrt{\frac{N_{2}^{2} d^{2}\beta_2}{2\lambda}z_\mathrm{eff}}$. Based on the gain function in \eqref{eqn_IIIB_3}, we present Theorems \ref{theorem1} and \ref{theorem2} to characterize the \unit[3]{dB} beamdepth and the \ac{EBRD}, respectively.
\begin{theorem} 
For a generalized \ac{URA}, the $\unit[3]{dB}$ beamdepth $\BD$ obtained by focusing a beam at a distance $\rf$ from the \ac{BS} is given by
\begin{equation} 
\BD = \frac{8\rf^2 {\RD} \alpha_\mathrm{\scalebox{0.5}{3dB}} \ \eta (\eta^2+1)\sqrt{\beta_1 \beta_2}} {\left[\eta \RD \sqrt{\beta_1 \beta_2}\right]^2 -\left[4\alpha_\mathrm{\scalebox{0.5}{3dB}} \rf (\eta^2+1)\right]^2}.
\label{eqn_IIIB_4}
\end{equation}
\label{theorem1}
\end{theorem}
\begin{proof}
Proof is provided in Appendix \ref{Appendix_A}. 
\end{proof}
\begin{theorem} 
The farthest distance for a given pair of azimuth and elevation angles, at which near-field beamfocusing for a \ac{URA} can be achieved, is termed the effective beamfocusing Rayleigh distance $\EBRD$ and is given by
\begin{equation}
\begin {aligned}
\rf < \frac{\eta\RD }{4\alpha_\mathrm{\scalebox{0.5}{3dB}} (1+\eta^2) }\sin \theta \sqrt{1-\sin^2 \theta \sin^2 \varphi}.
\end {aligned}
\label{eqn_IIIB_8}
\end{equation}
\label{theorem2}
\end{theorem}
\begin{proof}
Proof is provided in Appendix \ref{Appendix_B}.
\end{proof}
The beamdepth and \ac{EBRD} expressions for a \ac{USA} and a \ac{ULA} can be derived as follows:
\newtheorem{corollary}{Corollary}
\begin{corollary}
 For a \ac{USA}, beamdepth $\BD^{\mathrm{\scalebox{0.5}{USA}}}$ is obtained by substituting $\eta=1$ in \eqref{eqn_IIIB_4} to get
\begin{equation}
\BD^{\mathrm{\scalebox{0.5}{USA}}}= \begin{cases}\frac{16 \rf^{2} \RD \alpha_\mathrm{\scalebox{0.5}{3dB}} \sqrt{\beta_1 \beta_2}} {{(\RD \sqrt{\beta_1 \beta_2})} ^{2}- (8 \alpha_\mathrm{\scalebox{0.5}{3dB}} \rf)^{2}}, & \rf<\frac{\RD}{8 \alpha_\mathrm{3dB}} \sqrt{\beta_1 \beta_2}, \\ \infty, & \rf \geq \frac{\RD} {8 \alpha_\mathrm{\scalebox{0.5}{3dB}}} \sqrt{\beta_1 \beta_2}.\end{cases}
\label{eqn_IIIC_1}
\end{equation}
 Note that in the boresight case $ \sqrt{\beta_1 \beta_2} =1$ and $ \alpha_\mathrm{\scalebox{0.5}{3dB}} = 1.25$, \eqref{eqn_IIIC_1} reduces to the beamdepth expression in [Eq.~23, \cite{bjornson2021primer}]. The boresight case of a URA was considered in \cite{10443535}.
\end{corollary}
 \begin{corollary}
 For a \ac{ULA} along y-axis, beamdepth $\BD^{\mathrm{\scalebox{0.5}{ULA}}}$ is given by 
\begin{equation} 
\BD^{\mathrm{\scalebox{0.5}{ULA}}} = \begin{cases}\frac{8 \alpha_\mathrm{\scalebox{0.5}{3dB}} \rf^2{\RD} \cos^2{(\varphi)}} {( \RD \cos^2{(\varphi))^2 -(4\alpha_\mathrm{\scalebox{0.5}{3dB}} \rf )^2}},& \rf<\frac{\RD}{4 \alpha_\mathrm{\scalebox{0.5}{3dB}}} \cos^2{(\varphi)}, \\ \infty, & \rf \geq \frac{\RD} {4 \alpha_\mathrm{\scalebox{0.5}{3dB}}}\cos^2{(\varphi)}. \end{cases}
\label{eqn_IIIC_2}
\end{equation}
The special case of a ULA was covered in \cite{10934779}. Beamdepth for \ac{UCA} is derived in \cite{hussain2025uniform}. 
\end{corollary}
We evaluate the impact of the approximation in \eqref{eqn_B4}(b), where the cross term is neglected, on the beamdepth expression in \eqref{eqn_IIIB_4} by comparing it with the numerically computed beamdepth. The numerical beamdepth is obtained by computing the array gain using the exact array response vector in \eqref{eqn_B3} and identifying the \unit[3]{dB} points of the resulting gain function. As shown in Fig.~\ref{fig3_BD_A_vs_N}, the analytical beamdepth closely matches the numerical results, with a slight deviation observed at very short focal distances ($<2D$). At larger distances, this deviation vanishes as the contribution of the neglected cross term diminishes. Moreover, since the \ac{EBRD} represents the maximum distance limit for beamfocusing, the approximation in \eqref{eqn_B4}(b) has a negligible impact on the \ac{EBRD}.
Note that the \ac{EBRD} is conceptually distinct from the \ac{ERD} in~\cite{10541333}. Specifically, the ERD defines the boundary at which the beamforming loss incurred when approximating the near-field channel with a far-field model exceeds a prescribed threshold. For a \ac{ULA}, the \ac{ERD} is given by $0.367\, r_{\mathrm{RD}}\cos^{2}(\varphi)$ for a threshold of $0.05$. As shown in Fig.~\ref{fig4_ERD_vs_EBRD}, the \ac{ERD} overestimates the near-field region from a beamfocusing perspective, whereas only the region within the \ac{EBRD} truly supports beamfocusing.

\subsection{Impact of Spatial Focus Region and Array Geometry} \label{sec-3b}
A narrower beamdepth in the near-field allows for serving more users within the same angular direction while reducing mutual interference. In the following, we analyze how both beamdepth and \ac{EBRD} vary with spatial focus region and array geometry.
\subsubsection{Spatial Focus Region} \label{sec-3b-1}
 In the far-field, the beamwidth broadens as the beam is steered away from the boresight [Fig. 1, \cite{11095387}]. Similarly, the beamdepth $\BD$ increases as the focal region moves farther away from the array, as indicated by the quadratic dependence on $\rf$ in \eqref{eqn_IIIB_4}. This can also be observed from the approximate expression in \eqref{eqn_IIID_1} (Section~\ref{sec-4a-1}) and from Fig.~\ref{fig12_codeword_correlation} (Section~\ref{sec-6}). In the principal plane\footnote{For wide arrays, azimuth is considered the principal plane, whereas for tall arrays, elevation serves as the principal plane.}, beamdepth is minimal at the boresight and expands toward the endfire directions. For wide arrays ($\eta \gg 1$), beamdepth is primarily influenced by the azimuth angle $\varphi$, whereas for taller arrays ($\eta \ll 1$), it is mainly governed by the elevation angle $\theta$. In the cross plane, the variation of the beamdepth is more complex as both the length and width of the aperture affect it. To exemplify, the variation of beamdepth with respect to the azimuth angle $\varphi$ for different \ac{URA} configurations is shown in Fig. \ref{fig5_BD_vs_angles}, where $\BD$ increases at larger azimuth angles for wide arrays. Likewise, as the elevation angle $\theta$ increases from the boresight $\pi/2$ to endfire direction $5\pi/6$, beamdepth $\BD$ increases significantly for tall arrays. However, in the cross plane, beamdepth decreases at large elevation angles for wide arrays. A similar trend for beamwidth in the angular domain is reported in [Fig. 2.16, \cite{hansen2009phased}].

\begin{figure}[t]
\centering
\includegraphics[width=1\columnwidth]{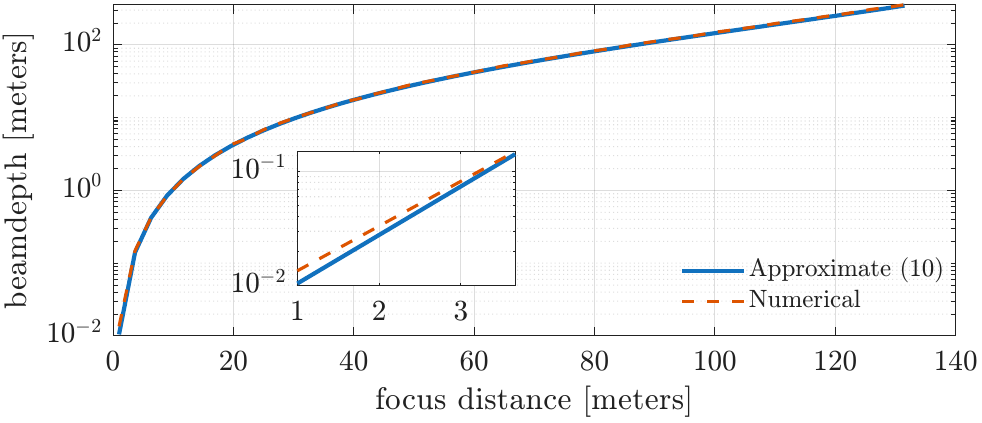}
\caption{Numerical and analytical beamdepth for a $(256 \times 16)$ \ac{URA} operating at $\unit[30]{GHz}$. The zoom inset highlights the beamdepth for focus distances less than $2D = 5$.}
\label{fig3_BD_A_vs_N}
\end{figure}

\begin{figure}[t]
\centering
\includegraphics[width=0.9\columnwidth]{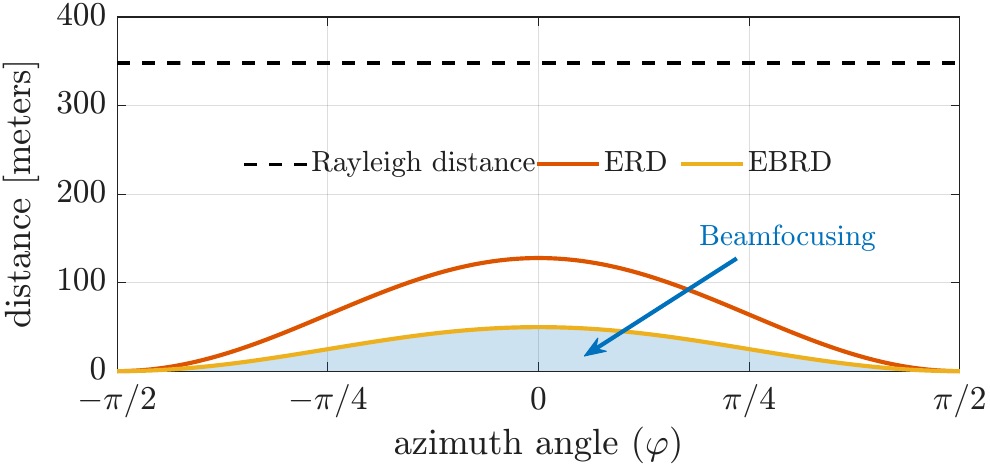}
\caption{Near-field boundary in terms of ERD and EBRD for a \ac{ULA}. Here $f_c = \unit[28]{GHz}$, $\NBS=256$, $\RD=\unit[348]{m}$.}
\label{fig4_ERD_vs_EBRD}
\end{figure}

\begin{figure}[t]
\centering
\includegraphics[width=1\columnwidth]{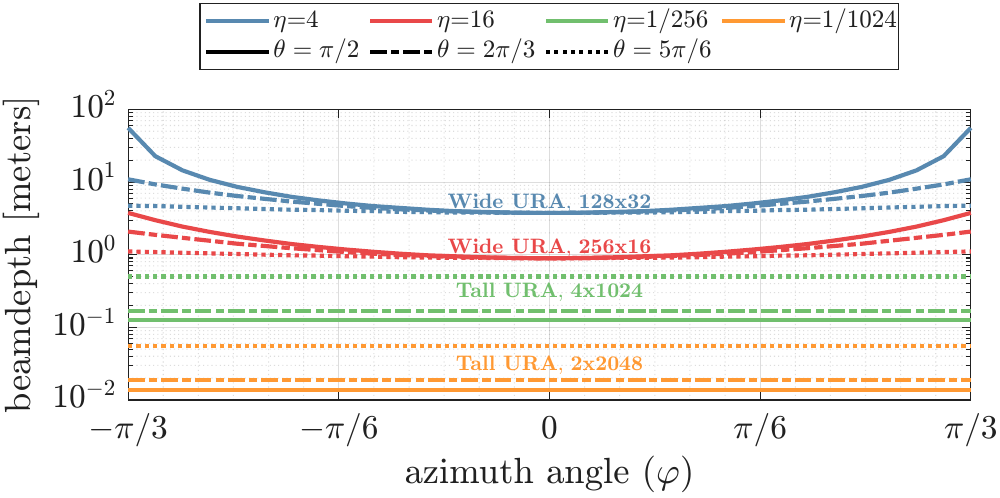}
\caption{Beamdepth with respect to azimuth angle $\varphi$ for different array configuration $\eta$ and elevation angle $\theta$.}
\label{fig5_BD_vs_angles}
\end{figure}
\subsubsection{Array Geometry} \label{sec-3b-2}
For a fixed number of antenna elements $\NBS = N_1 N_2$, different array geometries can be realized by varying the aspect ratio $\eta = N_1/N_2$, which results in different aperture lengths given by $D = d \sqrt{N_1^2 + N_2^2}$. Alternatively, one can fix the aperture length $D$ and vary $\eta$, allowing $\NBS$ to scale with the array shape. Furthermore, by keeping $D$ constant and varying the operating frequency, $\NBS$ can be adjusted as a function of both array geometry and wavelength. In the following analysis, we examine all three scenarios. To provide analytical insights, the beamdepth expression in \eqref{eqn_IIIB_4} can be approximated for boresight transmission as 
\begin{equation}
\BD \approx \frac{8\rf^2 \alpha_{\scalebox{0.5}{3dB}} (\eta^2 + 1)}{\eta \RD},
\label{eqn_BD_app}
\end{equation}
where $\rf$ remains the same for all the array configurations.

\textit{(i) Fixed Number of Antenna Elements:} 
For a fixed number of elements $\NBS = N_1N_2$, the aperture length $D = d \sqrt{N_1^2 + N_2^2}$ is minimized when the array is square, i.e., $N_1 = N_2$, corresponding to a \ac{USA}. As a result, the Rayleigh distance $\RD = \frac{2D^2}{\lambda}$ is minimized for square arrays and maximized for \ac{ULA} configurations. Additionally, the combined term $\alpha_{\scalebox{0.5}{3dB}}\left(\frac{\eta^2 + 1}{\eta}\right)$ achieves its maximum value when $\eta = 1$, i.e., for a square URA, and diminishes for other aspect ratios. Therefore, beamdepth $\BD$ is maximized when the array is square, as shown in Fig.~\ref{fig7_BD_Array} (black curve).

\begin{figure}[t!]
\centering
\includegraphics[width=0.9\columnwidth]{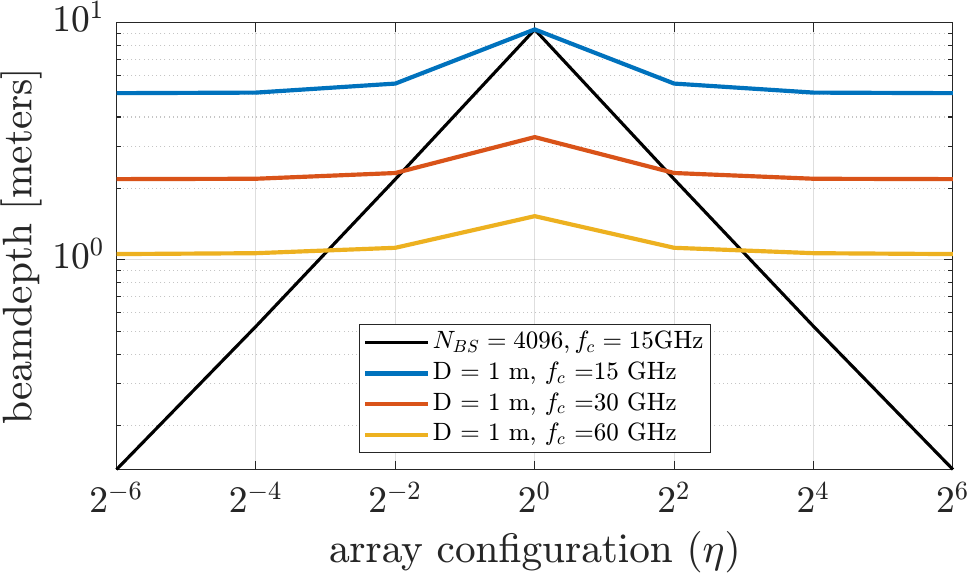}
\caption{Beamdepth with respect to $\eta$ for fixed number of antenna elements $\NBS$ and aperture length $D$.}
\label{fig7_BD_Array}
\end{figure}

\textit{(ii) Fixed Aperture Length:} 
For a fixed aperture length $D$, the Rayleigh distance $\RD$ remains constant across different aspect ratios. However, the composite factor $\alpha_{\scalebox{0.5}{3dB}} \left( \frac{\eta^2 + 1}{\eta} \right)$, which directly influences beamdepth, reaches its maximum for a square array. However, as shown in Fig.~\ref{fig7_BD_Array} (blue curve), the variation of beamdepth with respect to $\eta$ is less significant under a fixed aperture length than in the fixed element count scenario. This is because $\RD$—being relatively large in magnitude—dominates the beamdepth expression \eqref{eqn_BD_app}, making the influence of the geometry-dependent factor $\alpha_{\scalebox{0.5}{3dB}}\left(\tfrac{\eta^2 + 1}{\eta}\right)$ comparatively minor. 

Extending the Giga-\ac{MIMO} example discussed in Section~\ref{Introduction}, we further analyze the impact of the operating frequency on the beamdepth under a fixed aperture length $D$. Fig. ~\ref{fig7_BD_Array} presents the simulated beamdepth behavior for $D = 1~\text{m}$ when the carrier frequency increases from $\unit[15]{GHz}$ to $\unit[60]{GHz}$. The results demonstrate a clear decrease in beamdepth with increasing frequency. This trend is consistent with the approximation in \eqref{eqn_BD_app}, where the Rayleigh distance $\RD$ increases with frequency (due to the corresponding decrease in wavelength), thereby reducing the beamdepth even for a fixed aperture length. 
\subsubsection{Limits of Beamfocusing} \label{sec-3b-3}
The \ac{EBRD} defines the near-field boundary beyond which beamfocusing becomes infeasible. Within the \ac{EBRD} boundary, beamfocusing enables serving multiple users in the same angular direction, thereby enhancing multiuser capacity. It is therefore important to investigate methods for enlarging the \ac{EBRD}. We first analyze the case with a fixed element count, followed by the case with a fixed aperture length.

\textit{(i) Fixed Number of Antenna Elements:}
Based on \eqref{eqn_IIIB_8}, since the Rayleigh distance $\RD$ is minimal for a square \ac{URA}, and it dominates the other factor as explained earlier, the \ac{EBRD} is smallest for square \ac{URA} and increases as the \ac{URA} geometry transitions toward a \ac{ULA}. This relationship is illustrated in Fig. \ref{fig8_EBRD}, where the \ac{EBRD} is plotted with respect to azimuth angle $\varphi$, and as a function of $\eta$ and elevation angle $\theta$. Additionally, the \ac{EBRD} for wide arrays varies with the azimuth angle, while for tall arrays, it depends on the elevation angle. For all configurations in the principal plane, the \ac{EBRD} reaches its maximum at the boresight direction ($\varphi = 0,\theta = \pi/2 $) and decreases in the endfire directions. The largest \ac{EBRD} is observed for the \ac{URA} configuration with $\eta = 0.004$, while the smallest \ac{EBRD} is associated with the square \ac{URA} configuration ($\eta = 1$). Note that the \ac{EBRD} for $\eta = 0.016$ and $\eta = 0.004$ remains constant with respect to the azimuth angle $\varphi$, as these configurations represent tall \acp{URA} ($\eta \ll 1$). Furthermore, for these tall arrays, the \ac{EBRD} decreases further when the elevation angle deviates from the boresight ($\theta = \pi/2$).

\textit{(ii) Fixed Aperture Length:}
For a fixed aperture length, the square array continues to exhibit the lowest \ac{EBRD}, while the variation in \ac{EBRD} across different array geometries becomes less pronounced—consistent with the beamdepth behavior discussed earlier. Building on this, we now fix the array geometry and investigate how the \ac{EBRD} scales with aperture length $D$ and operating frequency. The objective is to identify the dominant factor that drives the expansion of the \ac{EBRD} region, thereby identifying the array geometries enabling more users to benefit from near-field communication.

To this end, we evaluate \ac{EBRD} as a function of azimuth angle for varying aperture sizes and carrier frequencies as shown in Fig.~\ref{fig8_EBRD_fixedeta}. At $D = 1$~m—which approximates the physical dimensions of the Giga-\ac{MIMO} array—\ac{EBRD} at boresight reaches approximately $13$~m at $15$~GHz. Increasing the frequency to $60$~GHz extends the \ac{EBRD} to around $55$~m, though this increase may still be modest. In contrast, increasing the aperture length offers a more significant impact: as shown in Fig. \ref{fig8_EBRD_fixedeta}, the \ac{EBRD} reaches approximately $123$~m at $15$~GHz when $D$ is increased to $3$~m. This dominant effect of aperture size over frequency can be attributed to the scaling behavior of the Rayleigh distance $\RD=\frac{2D^2}{\lambda}$, which increases quadratically with $D$ but only linearly with wavelength. Thus, enlarging the aperture length seems to be a more effective strategy for substantially expanding the near-field region. 

\subsection{Multiuser Capacity Analysis} \label{sec-3c}
In this subsection, we present numerical examples to evaluate multiuser capacity with respect to the \ac{EBRD} and the array geometry. We consider a hybrid beamformer with analog precoder $\mathbf{W}$. Additionally, zero-forcing is applied as the digital precoder. The overall \ac{SE} is given by
\begin{equation}
R = \sum_{m=1}^{5} \log_2 \left( 1 + \frac{ p_m \left| \mathbf{h}_{m}^{\mathsf{H}} \mathbf{W} \mathbf{f}_{m} \right|^2}{\sigma_m^2 + \sum_{l \neq m} p_l\left| \mathbf{h}_{m}^{\mathsf{H}} \mathbf{W} \mathbf{f}_{l} \right|^2} \right),
\label{eqn34}
\end{equation}
where $p_m$ denotes power allocated to the $\nth{m}$ user and $\sigma_m^2$ denotes the noise variance. Furthermore, $f_m$ represents the $\nth{m}$ column of the digital precoder.

\begin{figure}[t]
\centering
\includegraphics[width=1\columnwidth]{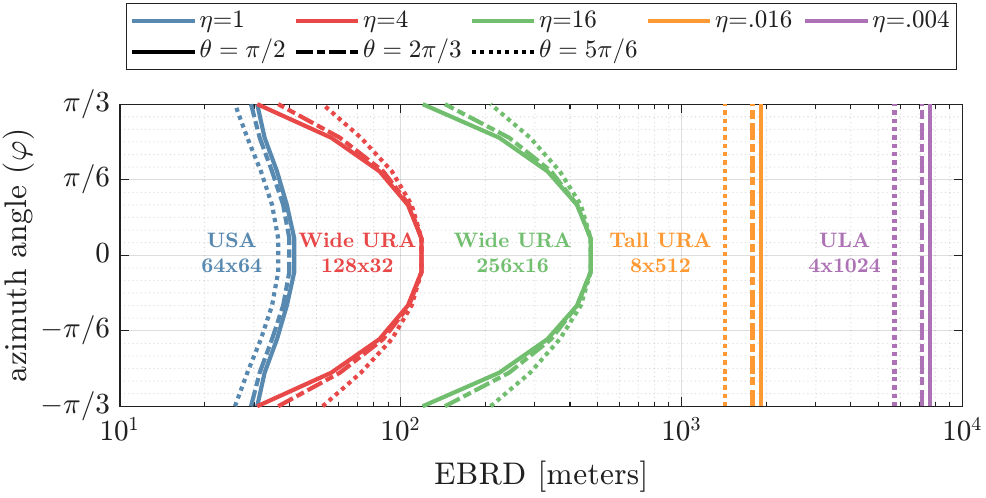}
\caption{Finite beamdepth limit in terms of \ac{EBRD} with respect to azimuth angle for different array configurations and elevation angles.}
\label{fig8_EBRD}
\end{figure}

\begin{figure}
\centering
\includegraphics[width=1\columnwidth]{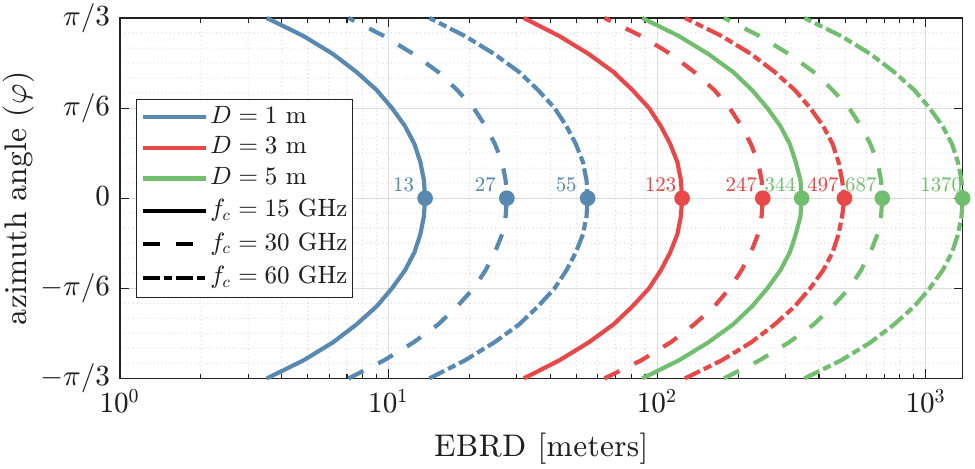}
\caption{Finite beamdepth limit in terms of \ac{EBRD} with respect to azimuth angle for different aperture lengths and frequencies.}
\label{fig8_EBRD_fixedeta}
\end{figure}

\subsubsection{SE vs. EBRD} \label{sec-3c-1}
Finite beamdepth is achievable only within the \ac{EBRD} region, limiting spatial multiplexing gains due to beamdepth to the near-field region bounded by the \ac{EBRD}. To illustrate this, we consider a $64 \times 8$ \ac{URA} serving five users positioned in the boresight direction $(\varphi=0, \theta = \pi/2)$. Users are distributed uniformly across three distinct regions: the \ac{EBRD} region $[2D, \EBRD]$, the near-field region beyond \ac{EBRD} $[\EBRD, \RD]$, and the far-field region $[\RD, 100\RD]$. We employ both polar and \ac{DFT} codebooks for each of these scenarios. The resulting \ac{SE} is shown in Fig. \ref{fig15_SE_vs_EBRD}. Within the \ac{EBRD} region, the polar codebook achieves a \ac{SE} of $5.8$ bps/Hz, compared to only $2$ bps/Hz with the \ac{DFT} codebook. Additionally, \ac{SE} is similar for both codebooks in the near-field region outside the \ac{EBRD} and in the far-field. It is important to emphasize that the \ac{EBRD} determines the near-field boundary in terms of multiplexing gain when \acp{UE} are distributed linearly along a single angular direction in the range domain. For uniformly distributed \acp{UE}, regions outside the \ac{EBRD} may still achieve slightly higher SE than the far-field case, owing to the improved lateral resolution at shorter distances compared to the far-field.

\subsubsection{SE vs. Array Geometry} \label{sec-3c-2}
Array geometry governs both beamdepth and the \ac{EBRD}; however, its impact diminishes under a fixed aperture length constraint. Nonetheless, under both fixed-element and fixed-aperture scenarios, elongated array configurations like \ac{ULA} yield the smallest beamdepth and the largest EBRD, thereby enhancing spatial multiplexing.

To quantify the impact of array geometry, we evaluate the \ac{SE}, as depicted in Fig. \ref{fig16_SE_vs_eta}, across $\eta = [1, 4, 16]$, under two constraints: fixed aperture length $D$ and fixed number of elements $\NBS$. As a baseline, the square \ac{URA} configuration ($\eta = 1$) is used, which is expected to yield the lowest \ac{SE} due to its wider beamdepth and limited \ac{EBRD}. In all scenarios, five users are randomly distributed along the boresight direction within the range $[2D, \EBRD]$, where $\EBRD$ varies with array geometry.

Under the fixed-element constraint ($\NBS = 1024$), \ac{SE} increases consistently with $\eta$, i.e., from $\eta = 1$ to $\eta = 16$, across all \ac{SNR} values. This trend aligns with Fig. \ref{fig7_BD_Array} (black curve), where a reduction in beamdepth with increasing $\eta$ (i.e., transitioning toward a \ac{ULA}) mitigates inter-user interference and enhances \ac{SE}.

Under the fixed-aperture constraint ($D = 2$), the \ac{SE} trend is more interesting. While the configuration with $\eta = 4$ achieves higher \ac{SE} than $\eta = 1$, the \ac{SE} for $\eta = 16$ drops compared to $\eta = 4$. This behavior can be explained by the \ac{SE} expression in \eqref{eqn34}, where the numerator is influenced by the beamforming gain, dependent on the number of antenna elements, and the denominator reflects mutual interference. Since beamdepths for $\eta = 4$ and $\eta = 16$ are nearly identical (as shown in Fig. \ref{fig7_BD_Array}), the variation in \ac{SE} arises primarily from the reduced element count in the $\eta = 16$ configuration, which diminishes the beamforming gain and thereby reduces the overall \ac{SE}.

\begin{figure}[t!]
\centering
\includegraphics[width=1\columnwidth]{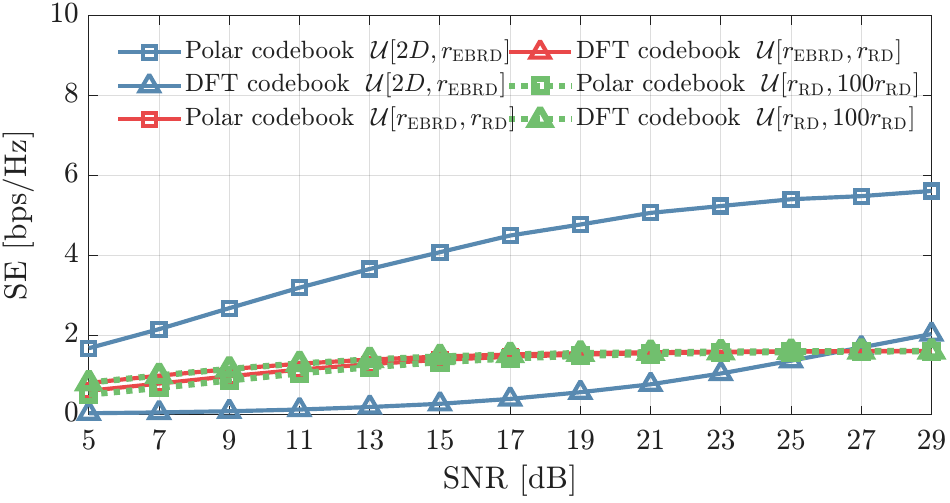}
\caption{ The comparison of \ac{SE} versus SNR, with respect to user distribution.}
\label{fig15_SE_vs_EBRD}
\end{figure}

\begin{figure}[t]
\centering
\includegraphics[width=1\columnwidth]{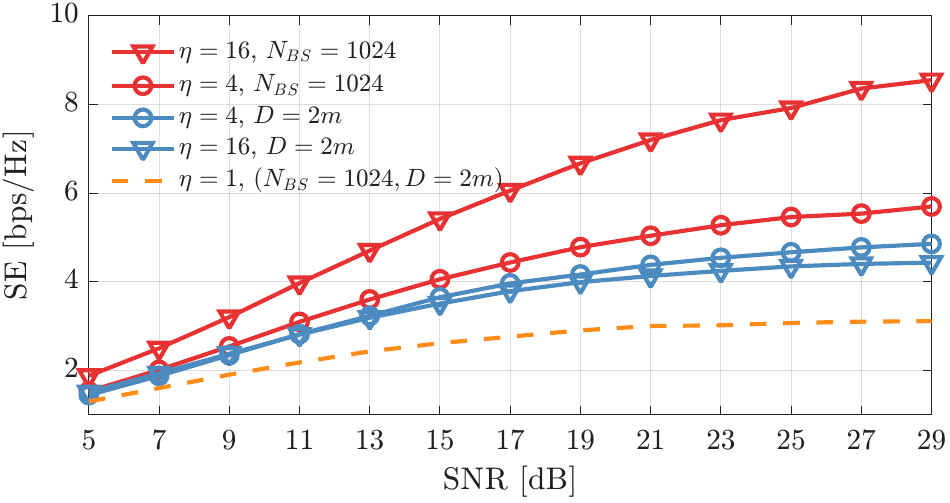}
\caption{ The comparison of \ac{SE} versus SNR, with respect to array configurations.}
\label{fig16_SE_vs_eta}
\end{figure}


\section{Axial vs. Lateral Beam Pattern Analysis} \label{sec-4}
Beam resolution is key to maximizing wireless communication capacity by enabling precise user separation and minimizing interference. In the near-field, beams can be focused in both the axial (range) and lateral (angular) dimensions, making it instructive to analyze and compare beam patterns in these dimensions. In the subsequent analysis, we examine the beam resolution, as well as the sidelobe patterns, in both the axial and lateral dimensions to provide a comparative analysis and implications for near-field wireless communication. 

\subsection{Beam Resolution} \label{sec-4a}
Axial and lateral resolution are determined by the beamdepth and beamwidth, respectively. For tractability, we compare beamdepth and beamwidth of a \ac{ULA} at a fixed distance $\rf$. Specifically, we consider a \ac{ULA} with boresight transmission, where $\alpha_\mathrm{\scalebox{0.5}{3dB}} = 1.75$ for $\varphi=0^\circ$.
\subsubsection{Axial Resolution} \label{sec-4a-1}
 For this configuration, the denominator in \eqref{eqn_IIIC_2} is simplified to $(\RD^2 -49\rf^2)$. A finite value of $\BD^{\mathrm{\scalebox{0.5}{ULA}}}$ is attained only when $\rf < \frac{\RD}{7}$, ensuring that $\RD^2$ remains significantly larger than $49\rf^2$. Consequently, we can neglect the $49\rf^2$ term and substitute $\RD = \frac{2D^2}{\lambda}$ in \eqref{eqn_IIIC_2}, yielding the following simplified expression for beamdepth of a \ac{ULA}: 
\begin{equation} 
 \BD^{\mathrm{\scalebox{0.5}{ULA}}} \approx 7\lambda {\left(\frac{\rf}{D} \right)}^2.
\label{eqn_IIID_1}
\end{equation}
\subsubsection{Lateral Resolution} \label{sec-4a-2}
The $\unit[3]{dB}$ beamwidth for $\NBS$ element \ac{ULA} can be approximated as 
$\varphi_\mathrm{\scalebox{0.5} {BW}} \approx \frac{\lambda}{\NBS d \cos(\varphi)}$ \cite{bjornson2024introduction}. Furthermore, lateral range resolution $\Delta r_\mathrm{az}$ in the azimuth direction is given by \cite{richards2005fundamentals}
\begin{equation}
 \Delta r_\mathrm{az} \approx \rf \varphi_\mathrm{\scalebox{0.5} {BW}}. 
 \label{eqn_IIIA_4}
\end{equation}
The lateral resolution based on the beamwidth $\varphi_\mathrm{\scalebox{0.5} {BW}}$ at $\varphi = 0^{\circ}$ can be obtained from \eqref{eqn_IIIA_4} as
\begin{equation} 
 \Delta r_{\mathrm{az}}^{\mathrm{\scalebox{0.5}{ULA}}} \approx \lambda {\left(\frac{\rf}{D} \right)},
\label{eqn_IIID_2}
\end{equation}
where $D =\NBS d$ represents aperture of the \ac{ULA} here. In the radiative near-field $\rf > 2D$, therefore, the factor $\frac{\rf}{D} > 1$. 
\subsubsection{Comparative Analysis}
Comparing \eqref{eqn_IIID_1} with \eqref{eqn_IIID_2}, $\Delta r_{\mathrm{az}}^{\mathrm{\scalebox{0.5}{ULA}}}$ yields finer resolution compared to the $\BD^{\mathrm{\scalebox{0.5}{ULA}}}$ that denotes the axial resolution. Therefore, for a given distance $\rf$, the range resolution in the lateral dimensions is consistently higher than that in the axial dimension. Notably, $\BD^{\mathrm{\scalebox{0.5}{ULA}}}$ varies quadratically with the factor $\frac{\rf}{D}$, while $\Delta r_{\mathrm{az}}^{\mathrm{\scalebox{0.5}{ULA}}}$ varies linearly. Therefore, as $\rf$ is increased, axial resolution degrades more rapidly than the lateral resolution, until \ac{EBRD} is reached, after which axial resolution capability is completely lost. Finally, it is pertinent to emphasize that lateral resolution is further enhanced in the near-field due to proximity, as the values of $\rf$ in $\Delta r_{\mathrm{az}}^{\mathrm{\scalebox{0.5}{ULA}}}$ are smaller in the near-field compared to those in the far-field. 

Extending the comparative analysis to planar arrays, we consider a \ac{URA}, which provides lateral resolution along both azimuth and elevation. For a \ac{ULA}, the array length coincides with the aperture and determines both the lateral and axial resolutions. In contrast, for a \ac{URA}, the axial resolution is governed by the maximum array extent (aperture length), while the lateral resolution depends independently on the array extent along each dimension. Similar to \eqref{eqn_IIID_1}, the axial resolution of a \ac{URA} can be obtained by substituting $\RD = \frac{2D^2}{\lambda}$ into \eqref{eqn_BD_app}, yielding $\BD^{\mathrm{\scalebox{0.5}{URA}}} \approx 4\lambda \left(\frac{\rf}{D}\right)^2 \frac{\alpha_{\scalebox{0.5}{3dB}}(\eta^2 + 1)}{\eta}$. The lateral resolution along each dimension is set by the array extent: azimuth by $D_y = N_1 d_y$ and elevation by $D_z = N_2 d_z$. Consequently, the lateral resolution in azimuth or elevation follows the same expression as in \eqref{eqn_IIID_2}, with $D$ replaced by the array length along the respective dimension. More generally, for the lateral resolution along a specific dimension $D_i$, $i \in \{y,z\}$, to exceed the axial resolution, the condition $\Delta r_{\mathrm{az}}^{\mathrm{\scalebox{0.5}{URA}}} < \BD^{\mathrm{\scalebox{0.5}{URA}}}$ must be satisfied, which leads to
\begin{equation}
\frac{D^2}{D_i} < \frac{4 \rf \, \alpha_\mathrm{3dB} (\eta^2 + 1)}{\eta}.
\label{lat_vs_axial_URA}
\end{equation}
To illustrate the utility of \eqref{lat_vs_axial_URA}, we consider a \ac{USA} with $\eta = 1$, $\alpha_\mathrm{3dB} = 1.25$, and $D = \sqrt{2}D_i$. Substituting these values yields the condition $\rf > \frac{D}{7}$, indicating that when the focal distance exceeds $\frac{D}{7}$, the lateral resolution of a \ac{USA} surpasses the axial resolution. Notably, the minimum distance for \ac{USW} operation in the near-field is $2D$, implying that this condition is always satisfied in the radiative near-field. Following a similar analysis for elongated \acp{URA}, it can be shown that at distances greater than $2D$, the lateral resolution remains superior to the axial resolution in the principal plane. However, in the cross-plane, the dimension with superior resolution is determined by the ratio of the aperture length to the array width or height.

\subsection{Sidelobe Levels} \label{sec-4b}
The array gain for a \ac{ULA} in the angular domain is $\mathcal{A}(u) = \left| \frac{\sin\left(\frac{\NBS}{2} u\right)}{\NBS \sin\left(\frac{u}{2}\right)} \right|^2$, where $u = \nu d \sin(\varphi)$ [Eq.~12, \cite{10403506}], while in the range-domain it is $\mathcal{G}(\gamma)=\abs{ \frac{C^2(\gamma)+S^2(\gamma)}{\gamma^2}}$ \cite{10934779}. In the near-field region, the focused beam exhibits both lateral sidelobes and axial forelobes. The \ac{PSL} represents the highest sidelobe amplitude relative to the mainlobe peak, indicating robustness against nearby narrowband interference. In the angular domain, the \ac{PSL} approximately occurs at $u \approx \frac{3\pi}{\NBS}$, yielding $\mathrm{PSL}_{\text{ang}} \approx 10 \log_{10} \left( \frac{1}{3 \pi / 2} \right)^2 \approx -13.46~\mathrm{dB}$, where a \textit{small angle} approximation is applied to obtain $\sin(\frac{3\pi}{2\NBS})\approx \frac{3\pi}{2\NBS}$. In the range domain, solving numerically $\frac{d}{d\gamma} \left( \frac{C^2(\gamma) + S^2(\gamma)}{\gamma^2} \right) = 0$ yields $\gamma_{\mathrm{PSL}} \approx 2.28$ with $10 \log_{10}(\mathcal{G}(\gamma_{\mathrm{PSL}})) \approx -8.7~\mathrm{dB}$. Note that the \ac{PSL} along the axial dimension is higher than that along the lateral dimension. Hence, suppression of axial forelobes may be more critical for mitigating interference.

\begin{figure}[!t]
 \centering
\begin{subfigure}[v]{1\columnwidth}
\centering
\includegraphics[width=1\columnwidth]{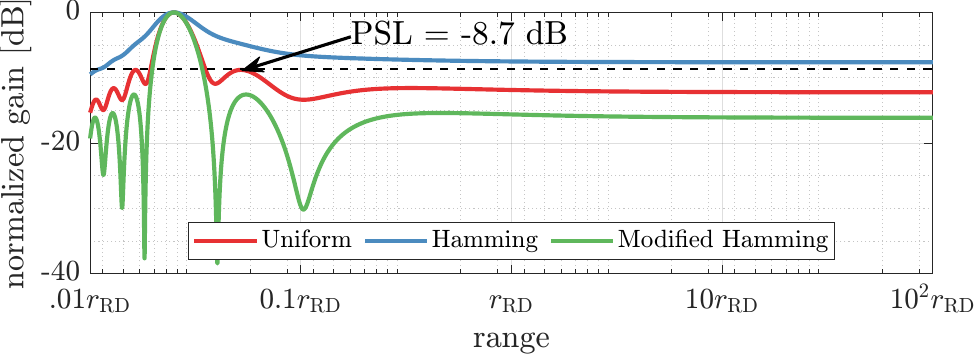}
\caption{}
 \label{7a}
 \end{subfigure}
 \hfill
 \begin{subfigure}[v]{1\columnwidth}
 \centering
\includegraphics[width=1\columnwidth]{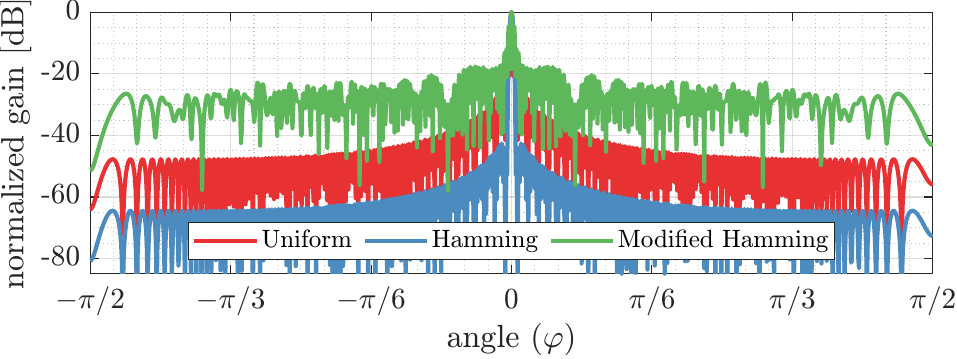}
 \caption{}
 \label{7b}
 \end{subfigure}
 \setlength{\belowcaptionskip}{-15pt}
\setlength{\abovecaptionskip}{0pt}
\caption{(a) Forelobes appear in the axial dimension, (b) while sidelobes are observed in the lateral dimension. Applying Hamming weights reduces the sidelobes in the lateral pattern, whereas modified Hamming weights suppress the forelobes in the axial pattern.}
\label{fig4_SLL}
\end{figure} 

Standard far-field techniques typically apply amplitude tapering distribution $\mathbf{f}(n)$ in the array domain as $\mathbf{b}(\varphi,r) \ \odot \ \mathbf{f}(n)$, to reduce the sidelobe levels in the angular domain. It may be desirable to utilize the same amplitude distribution functions to reduce axial forelobes as well. However, directly using these functions amplifies the forelobes in the axial dimension, since the Fresnel function in the range domain is quadratic in $n$ and introduces angular spreading. Here, we present a transformation that maps the standard far-field window $\mathbf{f}(n)$ to a modified window ${\mathbf{g}}(n)$, thereby allowing existing window designs to be directly adapted for the suppression of axial forelobes.
\begin{theorem}
Given a far-field window function $\mathbf{f}(n)$ designed to attenuate angular sidelobes, a corresponding distribution $\mathbf{g}(n)$ can be derived to suppress axial forelobes, via the transformation:
\begin{equation}
 \mathbf{g}(n) = |n| \, \mathbf{f}(n^2),
 \label{eqn_IIID_3}
\end{equation} \vspace{-2em}
\label{Theorem-4}
\end{theorem}
\begin{proof}
 The proof is inspired by \cite{graham1983analysis} and is in Appendix \ref{Appendix_C}.
\end{proof}
To demonstrate the application of \eqref{eqn_IIID_3}, Fig.~\ref{fig4_SLL} compares the axial and lateral beam patterns for three window designs: (i) uniform (untapered) case, (ii) conventional Hamming window $\mathbf{f}(n)$, and (iii) modified Hamming window ${\mathbf{g}}(n)$ obtained via \eqref{eqn_IIID_3}. The analysis considers a $256$-element \ac{ULA} operating at $\unit[28]{GHz}$, focused at $\rf = \RD/40$ and $\varphi=0^\circ$ in the near-field region. It is important to note that \eqref{eqn_IIID_3} can be applied to other window functions as well; the Hamming window is chosen here for illustration purposes, and is expressed as \cite{blackman1958measurement} 
\begin{equation}
\mathbf{f}(n) = 0.54 - 0.46 \cos\left( \frac{2 \pi n}{\NBS-1} \right), \quad 0 \leq n \leq \NBS-1.
\end{equation}
Furthermore, the modified Hamming distribution is obtained by transforming $\mathbf{f}(n)$ to $\mathbf{g}(n)$ according to \eqref{eqn_IIID_3}. Then the amplitude distribution $\mathbf{g}(n)$ is applied to the near-field array response vector as 
\begin{equation}
\mathbf{b}(\varphi,r)= \mathbf{b}(\varphi,r) \odot \mathbf{g}(n).
 \label{eqn_IIID_4}
\end{equation}
In Fig. \ref{fig4_SLL}, the conventional Hamming distribution reduces lateral sidelobes but increases the axial forelobes. The modified Hamming $\mathbf{g}(n)$ reduces the axial forelobes from $-8.7$ dB to $-13$ dB, but raises the lateral sidelobes. The desired amplitude distribution to suppress the axial forelobes degrades the angular pattern by increasing the lateral sidelobes, and vice versa. Therefore, achieving low sidelobe levels in both axial and lateral dimensions requires a suitable compromise between the two. Moreover, while the conventional Hamming window broadens the beamwidth, the modified Hamming window primarily increases the beamdepth without affecting the beamwidth. Beam broadening of the mainlobe occurs only when energy is redistributed from the sidelobes into the mainlobe. Since the modified Hamming window does not suppress sidelobe energy in the lateral dimension, no lateral energy redistribution occurs, and thus the beamwidth and lateral resolution remain unchanged. 

The overall lateral beam pattern is governed by the product of the antenna element pattern and the array factor; consequently, both the achievable beamwidth and the lateral sidelobe levels of the array are also affected by the beamwidth and sidelobe characteristics of the individual antenna elements. This principle, however, does not extend to the axial dimension. Beamdepth is not an intrinsic electromagnetic property of a single antenna element, and therefore, axial forelobes cannot be mitigated through element-level design. To simultaneously suppress sidelobes in both the axial and lateral domains, one can formulate an optimization problem that maximizes the two-dimensional mainlobe energy relative to the sidelobe energy of the beampattern. This leads to an energy concentration formulation, whose solution is obtained via a generalized eigenvalue problem \cite{ahmed_slepian}. As an alternative, sidelobe reduction may also be achieved using iterative projection-based techniques, such as the Gerchberg--Saxton algorithm applied in the polar domain \cite{10365224}.

\section{Spatial Degrees of Freedom} \label{sec-5}
Given a \ac{LoS} \ac{MIMO} channel $\mathbf{H} \in \mathbb{C}^{N_t \times N_r}$ with $N_t$ transmit and $N_r$ receive antennas, the spatial \ac{DoF} denote the number of independent sub-channels that support parallel communication modes. Mathematically, this corresponds to the number of non-zero singular values from the \ac{SVD} of $\mathbf{H}$, or equivalently, the rank of the Gramian matrix $\mathbf{R} = \mathbf{H}^{\mathrm{H}}\mathbf{H}$. For a far-field \ac{LoS} \ac{MIMO} channel, the rank is one; however, a higher rank can be achieved by leveraging the radiative near-field region. It is important to note that channel capacity depends not only on the rank but also on the condition number of $\mathbf{H}$, which ideally equals one. Consequently, the \ac{EDoF} is employed in practice to quantify the number of useful sub-channels at high \ac{SNR}, where multiple \ac{EDoF} contribute to capacity. Having many \ac{EDoF} is only beneficial if the \ac{SNR} is sufficiently large to exploit them. In contrast, at low \ac{SNR}, channel capacity is dominated by the largest singular value, and a rank-one channel is preferable.

In this section, we first review existing metrics for \ac{EDoF}. We then formulate a new expression for \ac{EDoF} tailored to the \ac{URA}, characterizing its dependence on distance and angle. Finally, we provide comparisons across different \ac{URA} configurations and establish the boundary where the \ac{EDoF} reduces to one.

In practice, the singular values of $\mathbf{H}$ remain approximately constant up to a certain index, beyond which they decay rapidly toward zero. This transition index is defined as the \ac{EDoF} \cite{ouyang2023near}, expressed as $\sigma_1 \approx \sigma_2 \approx \cdots \approx \sigma_{\text{EDoF}} \gg \sigma_{\text{EDoF}+1} > \cdots > \sigma_{\text{DoF}}$. Accordingly, $\operatorname{tr}(\mathbf{R}) = \sum_{n=1}^{\text{DoF}} \sigma_n^2 
\;\approx\; \text{EDoF} \times \sigma_{\text{EDoF}}^2,$ and $\|\mathbf{R}\|_\mathsf{F}^2 = \sum_{n=1}^{\text{DoF}} \sigma_n^4 \;\approx\; \text{EDoF} \times \sigma_{\text{EDoF}}^4.$ By combining these two relations, the \ac{EDoF} can be approximated as
\begin{equation}
\text{EDoF}_1 \;\approx\; 
\frac{\big(\sum_{n=1}^{\text{DoF}} \sigma_n^2\big)^2}
{\sum_{n=1}^{\text{DoF}} \sigma_n^4}
\;=\; \frac{\operatorname{tr}^2(\mathbf{R})}{\|\mathbf{R}\|_\mathsf{F}^2},
\label{eqn_EDoF1}
\end{equation}
which quantifies the number of dominant singular values and thus, the effective number of useful sub-channels. Although compact in form, this expression does not directly reveal how \ac{EDoF} vary with distance and angle in the near-field region.
\subsection{Effective Spatial Degrees of Freedom} \label{sec-5a}
The \ac{DoF} correspond to the number of independent columns, whereas the \ac{EDoF} represent the set of mutually orthogonal columns in $\mathbf{H}$. Our proposed approach transforms $\mathbf{H}$ into the angular domain via the \ac{DFT} to identify the set of orthogonal sub-channels. In this domain, the \ac{EDoF} can be interpreted as the number of distinguishable transmit beams observable at the receiver. From a geometric perspective, it can further be approximated as the ratio between the receive aperture length $D_r$ and the $3$ dB beamwidth of the transmit array given in \eqref{eqn_IIID_2}. Consequently, the \ac{EDoF} is expressed as
\begin{equation}
    \text{EDoF}_2 \approx \frac{D_r}{\tfrac{\lambda r}{D_t}} 
    \;=\; \frac{D_r D_t}{\lambda r}.
    \label{eqn_EDoF2}
\end{equation}
The above expression is also derived in \cite{Miller:00} based on optical diffraction theory. Direct computation of the \ac{DFT} of $\mathbf{H}$ is intractable due to the coupling among the \ac{DoF} of the \ac{MIMO} channel. To address this, we first estimate the \ac{EDoF} based solely on the transmit array and subsequently rescale the result by the receive aperture length. This formulation is motivated by \eqref{eqn_EDoF2}, where the \ac{EDoF} for a unit-length aperture is first obtained from the transmit beamwidth and then scaled by $D_r$. In the following, we describe the proposed procedure for computing the \ac{EDoF} based on the transmit array.

The spherical wavefront can be viewed as a weighted sum of planar wavefronts. More specifically, the near-field array response vector in \eqref{eqn_B3} can be degenerated into a series of \ac{DFT} beams \cite{10934778}. For a \ac{URA}, orthogonal \ac{DFT} beams can be obtained by uniformly sampling the directional cosines $u_y$ and $u_z$ as 
\begin{subequations}
\begin{equation}
u^{(n_1)}_y =\sin\theta_\mathrm{n_1} \sin\varphi_\mathrm{n_1}= -\frac{n_1\lambda}{2N_1d},
\label{eqn_IV_2}
\end{equation}
\begin{equation}
u^{(n_2)}_z = \cos\theta_\mathrm{n_2}= -\frac{n_2\lambda}{2N_2d}. 
\label{eqn_IV_3}
\end{equation}
\end{subequations}
where $(u^{(n_1)}_y,u^{(n_2)}_z)$ represent spatial frequency resource for a \ac{URA}. Furthermore, to ensure that the sampled spatial frequencies correspond to physically realizable propagating waves, $u_y$ and $u_z$ must be real-valued and lie within the visible region, satisfying the unity constraint $\sqrt{u_y^2 + u_z^2} \leq 1$ \cite{van2002optimum}. The $\nth{(n_1,n_2)}$ element of the \ac{DFT} vector can be represented as 
\begin{equation} 
 \ar(u^{(n_1)}_y,u^{(n_2)}_z) =e^{j 
\nu \left({n_1}du^{(n_1)}_y +{n_2}du^{(n_2)}_z\right)}. 
\label{eqn_V_2}
\end{equation}
We compute the gain function $f(u^{(n_1)}_y,u^{(n_2)}_z; \varphi,\theta,\rf)$ associated with each spatial frequency $(u^{(n_1)}_y,u^{(n_2)}_z)$ by computing the \ac{DFT} coefficients for the near-field response vector $\br^{\mathsf{H}}(\varphi,\theta,\rf)$ as given in \eqref{eqn_V_3}.

\begin{figure*}[b]
\begin{equation}
\begin{aligned}
& f(u^{(n_1)}_y,u^{(n_2)}_z; \varphi,\theta,\rf) 
 \overset{\Delta}{=}\left|\br^{\mathsf{H}}(\varphi,\theta,\rf) \ar(u^{(n_1)}_y, u^{(n_2)}_z)\right|^2, \\
& {\approx} \left|\frac{1}{N_1 N_2} \sum_{n_{1}} \sum_{n_{2}} e^{-j 
\frac{2\pi}{\lambda}\left(- n_1 d u_y - n_2 d u_z + \tfrac{n^{2}_{1} d^{2}}{2\rf} \beta_1 + \frac{n_2^{2} d^{2}}{2\rf} \beta_2 \right) e^{j \frac{2\pi}{\lambda}\left(n_1 d u^{(n_1)}_y +n_2 d u^{(n_2)}_z \right)}
 }\right|^2, \\
& {\approx} \left| \frac{1}{N_1} \sum_{-N_{1} / 2}^{N_{1} / 2} e^{-j \pi \left(- n_1(u_y - n_1 u^{(n_1)}_y ) + \frac{n_1^{2} d}{2\rf} \beta_1 \right)} \frac{1}{N_1} \sum_{-N_{2} / 2}^{N_{2} / 2} e^{-j \pi \left(- n_2 (u_z - n_2 u^{(n_2)}_z) + \frac{n_2^{2} d}{2\rf} \beta_2 \right)}
 \right|^2. 
\end{aligned}
\label{eqn_V_3}
\end{equation}
 \end{figure*}
\begin{theorem} 
For a \ac{URA}, the normalized gain function $f(u^{(n_1)}_y,u^{(n_2)}_z; \varphi,\theta,\rf)$ can be approximated in terms of Fresnel functions as
\begin{equation}
\begin{aligned}
\tilde{f}\left(u^{(n_1)}_y,u^{(n_2)}_z; \varphi,\theta,\rf\right)\approx \frac{\mathcal{G}_{1}\left(\gamma_{1},\gamma_{2} \right) \mathcal{G}_{2}\left(\gamma_{3}, \gamma_{4}\right)}{16\gamma_2 \gamma_4},
\end{aligned}
\label{eqn_V_4}
\end{equation}
where
\begin{equation}
\begin{aligned}
 &\gamma_{1}=(u^{(n_1)}_y-u_y) \sqrt{\frac{\rf}{d\left(1-u^{2}_y\right)}}, \gamma_{2}=\frac{N_1}{2} \sqrt{\frac{d\left(1-u_{y}^{2}\right)}{\rf}}, \\
&\gamma_{3}=(u^{(n_2)}_z-u_z) \sqrt{\frac{\rf}{d\left(1-u^{2}_z\right)}}, \gamma_{4}=\frac{N_2}{2} \sqrt{\frac{d\left(1-u_{z}^{2}\right)}{\rf}}.
\end{aligned}
\label{eqn_V_5}
\end{equation}
\label{theorem_3}
\end{theorem}\vspace{-1em} 
Moreover, $\mathcal{G}_{1}\left(\gamma_{1}, \gamma_{2}\right) {=}[C\left(\gamma_{1}+\gamma_{2}\right)-C\left(\gamma_{1}-\gamma_{2}\right)]^2 + [S\left(\gamma_{1}+\gamma_{2}\right)-S\left(\gamma_{1}-\gamma_{2}\right)]^2$ and $\mathcal{G}_{2}\left(\gamma_{3}, \gamma_{4}\right) {=}[C\left(\gamma_{3}+\gamma_{4}\right)-C\left(\gamma_{3}-\gamma_{4}\right)]^2 + [S\left(\gamma_{3}+\gamma_{4}\right)-S\left(\gamma_{3}-\gamma_{4}\right)]^2$, where $C(\cdot)$ and $S(\cdot)$ are the Fresnel integrals, defined as $C(x)=\int_{0}^{x} \cos \left(\frac{\pi}{2} t^{2}\right) \mathrm{d} t$ and $S(x)=\int_{0}^{x} \sin \left(\frac{\pi}{2} t^{2}\right) \mathrm{d} t$.
 \begin{proof}
 Proof is provided in Appendix \ref{Appendix_D}.
\end{proof}
Given user location $\left(\varphi,\theta,\rf \right)$, the gain function $\tilde{f}\left(u^{(n_1)}_y,u^{(n_2)}_z; \varphi,\theta,\rf\right)$ in \eqref{eqn_V_4} depends on $\left(u^{(n_1)}_y,u^{(n_2)}_z\right)$. Note that $u^{(n_1)}_y$, and $u^{(n_2)}_z$ are contained in $\gamma_1$ and $\gamma_3$ respectively. Furthermore, $\mathcal{G}_1(\gamma_1) = \mathcal{G}_1(-\gamma_1)$ and $\mathcal{G}_2(\gamma_3) = \mathcal{G}_2(-\gamma_3)$. To characterize the bandwidth of spatial frequencies, we adopt the $\unit[3]{dB}$ threshold criterion. The $\unit[3]{dB}$ criterion is consistent with the conventional $\unit[3]{dB}$ definitions of beamwidth \cite{sherman1962properties} and beamdepth \cite{bjornson2021primer}. This threshold ensures that the retained spatial paths correspond to sufficiently strong and physically meaningful propagation modes. Moreover, the concept can be easily extended to an $x$-dB criterion for other values of $x$. Alternative thresholds may affect the absolute numerical value of the \ac{EDoF}; however, the qualitative trends and the relative comparisons across different array geometries remain unchanged, as presented herein. We count the spatial frequencies with the normalized gain above $0.5$ and term it as the \ac{EDoF} given by
\begin{equation}
\mathrm{EDoF}_3(\varphi,\theta,\rf) 
= \sum_{n_1} \sum_{n_2} 
\mathbf{1}\!\left\{
\tilde{f}\!\left(u^{(n_1)}_y,u^{(n_2)}_z;\varphi,\theta,\rf\right) \geq \tfrac{1}{2}
\right\}.
\label{eqn_EDoF3}
\end{equation}
where $\mathbf{1}(\cdot)$ denotes the indicator function. For a far-field \ac{UE}, a single peak value for $\tilde{f}$ is obtained. In contrast, for a near-field \ac{UE}, significant gain values for multiple spatial frequencies are observed around the user's spatial angles $(u_y, u_z)$. The $\mathrm{EDoF}_3$ is derived from a \ac{DFT}-based representation and therefore offers a coarse spatial resolution. To enhance the resolution, $\mathrm{EDoF}_3$ can be computed numerically using an oversampled FFT of size $\mathcal{K}\NBS$, where $\mathcal{K}$ denotes the oversampling factor. The resulting \ac{EDoF} is then normalized by $\mathcal{K}$ to compensate for the oversampling. We will now analyze \eqref{eqn_EDoF3} with respect to the focus region and array configuration.

\begin{figure}[t]
\centering
\includegraphics[width=1\columnwidth]{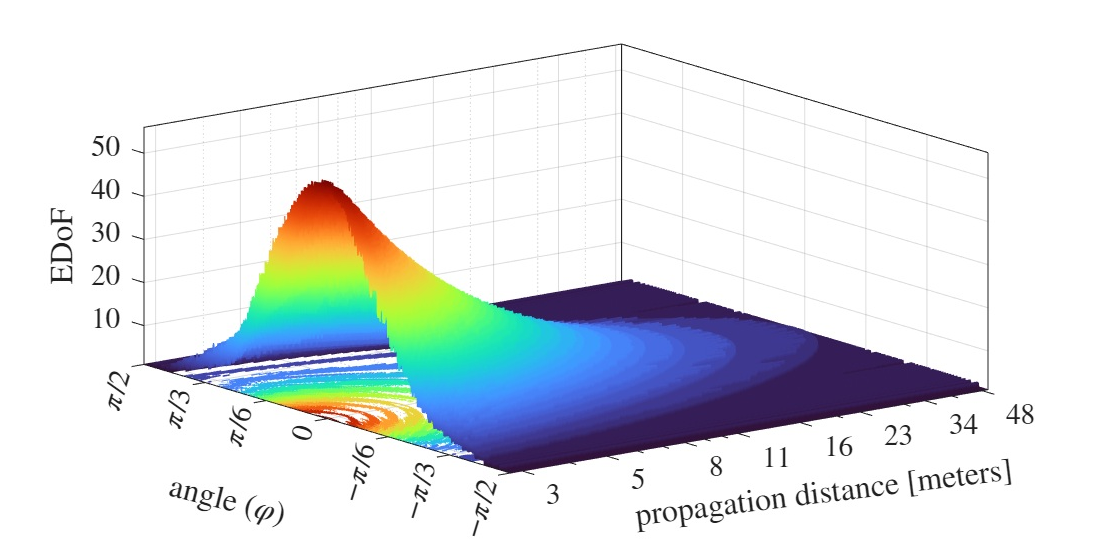}
\caption{ \ac{EDoF} for a \ac{ULA} with respect to near-field communication distance and angle.}
\label{fig10_DOF_range_angle}
\end{figure}

\begin{figure}[t]
\centering
\includegraphics[width=1\columnwidth]{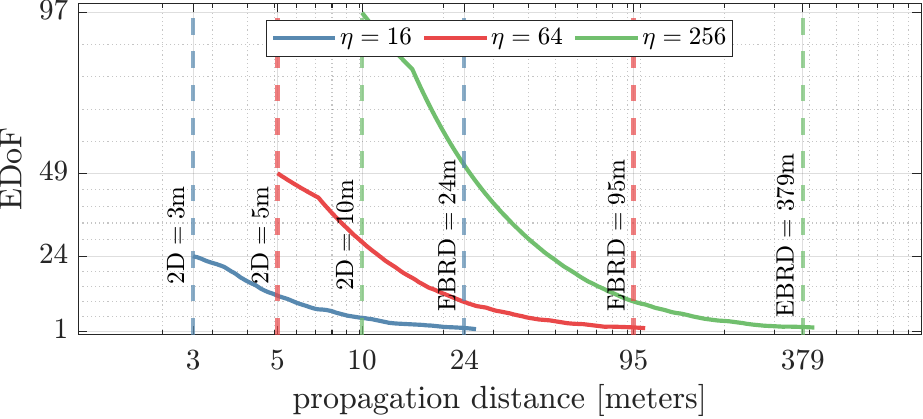}
\caption{ \ac{EDoF} versus near-field communication distance for different array configurations.}
\label{fig11_DOF_vs_AC}
\end{figure}

\subsubsection{\ac{EDoF} vs. Focus Region} \label{sec-5a-1}
To gain further insight, we plot the \ac{EDoF} as a function of near-field communication ranges for a \ac{ULA} at $\rf \in [2D \ \frac{\RD}{7}]$ and azimuth angles $\varphi \in [-\pi/2 \ \pi/2$], as shown in Fig. \ref{fig10_DOF_range_angle}. The maximum \ac{EDoF} is observed when the user's range is $\rf=2D$ at the boresight angle $\varphi=0$. The \ac{EDoF} decreases monotonically as the communication distance increases. Similarly, the \ac{EDoF} diminish as the user's azimuth angle moves away from the boresight. The minimum \ac{EDoF}, equal to one, occurs at $\frac{\RD}{7}\cos^2\varphi$, where $\RD = \unit[35]{m}$. Therefore, while considering the transmit array only, the \ac{EBRD} given in \eqref{eqn_IIIB_8} represents the boundary beyond which the \ac{EDoF} in the near-field and far-field become equivalent. 
\subsubsection{\ac{EDoF} vs. Array Configuration} \label{sec-5a-2}
The array configuration has a significant impact on the \ac{EDoF}. To illustrate this, Fig.~\ref{fig11_DOF_vs_AC} plots the \ac{EDoF} versus the user range for four array configurations, $\eta \in \{16, 64, 256\}$, with $N_t=1024$. For $\rf < 2D$, the \ac{EDoF} remains constant and equals the value at $2D$. Two important insights are worth noting here: (i) the \ac{EDoF} is maximized for a \ac{URA} with $\eta=256$ and minimized for a more square configuration with $\eta=16$, and (ii) for each configuration, the minimum \ac{EDoF} occurs at the \ac{EBRD} boundary, indicated by the vertical dashed lines.

\begin{figure}[t]
\centering
\includegraphics[width=1\columnwidth]{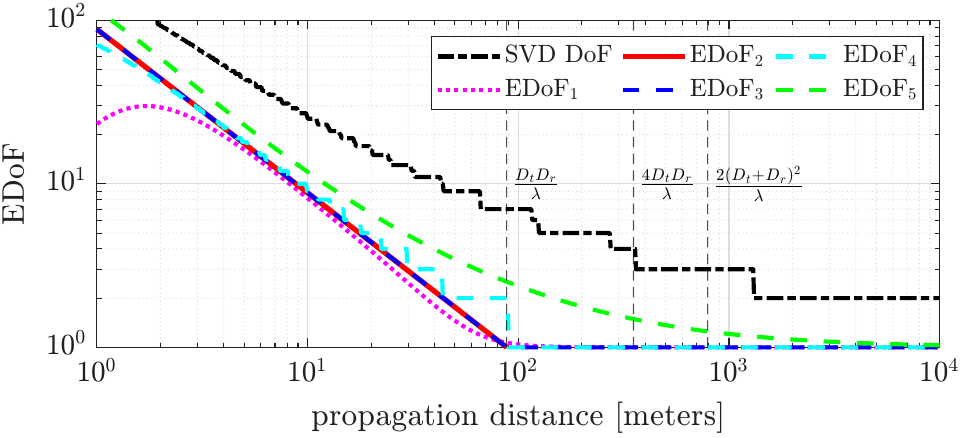}
\caption{\ac{EDoF} of a \ac{MIMO} system with a $256 \times 16$ \ac{URA} at the transmitter and a $128 \times 8$ \ac{URA} at the receiver, and $f_\mathrm{c} = \unit[28]{GHz}$.}
\label{fig_EDoF_comparison}
\end{figure}

\begin{figure}[t]
\centering
\includegraphics[width=1\columnwidth]{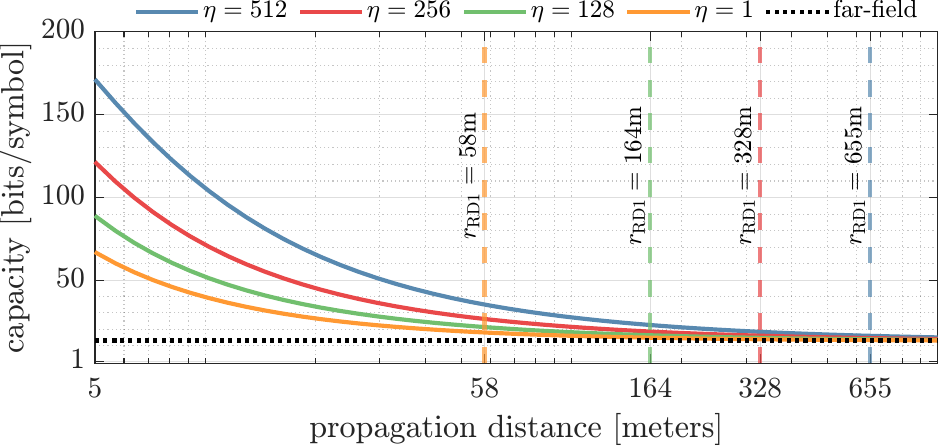}
\caption{ Comparison of LoS MIMO channel capacity versus propagation distance for different array configurations.}
\label{fig15_SE_MIMO_eta}
\end{figure}

\subsection{\ac{EDoF} for MIMO and Capacity Analysis} \label{sec-5b}
In this subsection, we present the rescaling of \ac{EDoF} for \ac{MIMO} and compare it with existing \ac{EDoF} formulations. To account for the receive aperture $D_r$, we combine \eqref{eqn_EDoF3} and \eqref{eqn_EDoF2} to obtain the following expression for the scaled distance $\hat{r}$
\begin{equation}
    \hat{r} = \frac{D_r D_t}{\lambda \, \text{EDoF}_3}.
    \label{eqn_rescaled_distance}
\end{equation}
Furthermore, the Rayleigh distance for \ac{MIMO} systems depends on both $D_t$ and $D_r$. We consider the following three expressions to characterize the \ac{MIMO} Rayleigh distance:
\begin{equation}
    \RDA = \frac{D_t D_r}{\lambda},
    \label{eqn_RD1}
\end{equation}
\begin{equation}
    \RDB = \frac{2(D_t + D_r)^2}{\lambda},
    \label{eqn_RD2}
\end{equation}
\begin{equation}
    \RDC = \frac{4 D_t D_r}{\lambda},
    \label{eqn_RD3}
\end{equation}
where \eqref{eqn_RD1} is presented in \cite{bjornson2024introduction}, \eqref{eqn_RD2} is based on the maximum allowable phase error of $\pi/8$ \cite{10819602}, and \eqref{eqn_RD3} is derived from the condition that the phase discrepancy between the near-field \ac{MIMO} channel and the channel modeled by the product of near-field array response vectors does not exceed $\pi/8$ \cite{10078317}. Out of these three metrics, we aim to identify the boundary at which the \ac{EDoF} for a \ac{MIMO} configuration reduces to one. To this end, Fig.~\ref{fig_EDoF_comparison} plots the \ac{EDoF} with respect to the propagation distance. In addition to the existing and the proposed \ac{EDoF} expressions, two alternative \ac{EDoF} formulations are derived from the singular values $\mathbf{s}_n$ of $\mathbf{H}$. First, by normalizing the singular values as $\mathbf{q}_n = \mathbf{s}_n / \max(\mathbf{s})$, we obtain $\text{EDoF}_{4} = \sum \mathbf{1}(\mathbf{q}_n > 0.5)$, which corresponds to a hard-thresholding criterion. Second, by normalizing with the sum, $\mathbf{p}_n = \mathbf{s}_n / \sum \mathbf{s}_n$, the \ac{EDoF} is computed using the entropy-based definition, $\text{EDoF}_{5} = \exp\!\left(-\sum \mathbf{p}_n \log(\mathbf{p}_n)\right)$, which reflects a soft-thresholding with $\text{EDoF}_5 \in \mathbb{R}^{+}$, whereas in contrast we obtain integer-valued $\text{EDoF}_{4} \in \mathbb{Z}^{+}$. As shown in Fig.~\ref{fig_EDoF_comparison}, the \ac{DoF} based on the \ac{SVD} tends to overestimate \ac{EDoF} and serves as an upper bound. This is because the \ac{SVD}-based \ac{DoF} reflects the algebraic rank of the channel and counts even weak spatial modes. In contrast, $\mathrm{EDoF}_3$ accounts only for energy-dominant and physically resolvable modes, yielding a smaller yet more realistic estimate. The proposed $\text{EDoF}_3$ and all other \ac{EDoF} definitions, except $\text{EDoF}_5$, converge to one at the $\RDA$ boundary, while $\RDB$ and $\RDC$ overestimate the near-field extent in terms of \ac{EDoF}. Since the proposed $\text{EDoF}_3$ and existing \ac{EDoF} definitions converge to one at the $\RDA$ boundary, we conclude that $\RDA$ provides the most consistent criterion for the distance at which the \ac{EDoF} effectively reduces to one.

We present a numerical example comparing the \ac{MIMO} capacity for different array configurations, approximated as $C \approx \text{EDoF} \, \log_2\!\big(1 + \rho / \text{EDoF}^2\big)$, where $\rho = \frac{PN_tN_r}{\sigma^2}$, and $P$ is the transmitted power. The receiver employs a \ac{URA} with fixed $\eta = 256 \times 4$ elements operating at $\unit[15]{GHz}$, while the transmitter varies $\eta \in \{512, 256, 128, 1\}$. In all cases, $N_t = N_r = 1024$, and the \ac{SNR} is set to $\unit[20]{dB}$. The resulting capacity, plotted in Fig.~\ref{fig15_SE_MIMO_eta}, shows that the \ac{URA} with $\eta = 512$ achieves the highest capacity across all distances. This superior performance arises from the enhanced \ac{EDoF}, which is maximized for elongated \acp{URA} compared to square \acp{URA} (see Fig.~\ref{fig11_DOF_vs_AC}). Moreover, for all configurations, the capacity within the $\RDA$ region significantly exceeds that of the far-field rank-one channel. The vertical dashed lines in Fig.~\ref{fig15_SE_MIMO_eta} indicate the $\RDA$ boundary for each configuration. We employ $\text{EDoF}_5$ to evaluate the capacity, which does not reduce to one at the $\RDA$ (as also observed in Fig.~\ref{fig_EDoF_comparison}), so the capacity at this point remains higher than the far-field value. $\text{EDoF}_4$ is not used because it relies on hard thresholding, which may allocate zero power to channels with singular values below the chosen threshold. While this work primarily focuses on \ac{LoS} propagation, near-field operation can also yield a higher \ac{EDoF} in \ac{NLoS} scenarios. The \ac{EDoF} of \ac{MIMO} channels under \ac{NLoS} conditions is strongly influenced by the geometric distribution of scatterers. In rich scattering environments, where scatterers span the full angular domain, the \ac{MIMO} channel can achieve full rank in both near-field and far-field regimes. In contrast, under sparse or limited scattering, the near-field may provide a higher channel rank due to its superior lateral resolution at shorter distances. However, large-scale arrays introduce spatial nonstationarity (SNS), manifested as significant power variations across the antenna elements. SNS primarily arises from partial blockage by finite-sized objects and from incomplete scattering, where objects do not interact uniformly with all antenna elements. The proposed DFT-based method for computing the \ac{EDoF} is not directly applicable under such conditions. In this case, the visibility regions of the array must first be identified to account for SNS effects, after which the DFT-based approach can be applied within each visible region.


\section{Codebook design and Channel Estimation} \label{sec-6}
The insights from Section \ref{sec-3} on beamdepth and \ac{EBRD}, combined with the sidelobe suppression explained in Section \ref{sec-4}, form the basis for constructing an efficient polar codebook. By leveraging these properties, the proposed codebook is specifically designed to optimize beamforming and \ac{CS}-based channel estimation in the radiative near-field. In this section, we introduce a polar codebook design tailored for a \ac{URA}, followed by a brief discussion on channel estimation. 

The \ac{DFT} codebook is commonly employed in the far-field to serve as a dictionary for \ac{CS}-based channel estimation and beam training. It comprises far-field array response vectors $\ar(\varphi,\theta)$ obtained by uniform sampling of the directional cosines. In that case, column coherence, which is defined as the maximum of the absolute inner products between two different columns of the dictionary matrix, is zero, since the columns of the \ac{DFT} matrix are mutually orthogonal. According to \ac{CS} theory, a small value of mutual coherence promotes sparsity and thus improves \ac{CS}-based channel estimation \cite{davenport2012introduction}. 

The near-field codebook $\boldsymbol{\Phi}$, also termed polar codebook, comprises potential near-field array response vectors. The criterion of optimal codebook construction is to minimize the column coherence $\mu$
\begin{equation}
\mu=\max _{p \neq q}\left|\br^{\mathsf{H}}\left(\varphi_{p}, \theta_{p}, r_{p}\right) \br\left(\varphi_{q}, \theta_{q}, r_{q}\right)\right|^2,
\label{eqn_IV_1}
\end{equation}
where $p$ and $q$ denote the column indices of $\boldsymbol{\Phi}$. The reduction in column coherence $\mu$ implies large angle and distance sampling intervals. On the other hand, the sampling intervals should be as small as possible to improve resolution and ensure complete coverage of the near-field region. Therefore, the key design considerations for constructing a polar-domain codebook include reducing column coherence, ensuring complete near-field coverage, and limiting the overall codebook size. Existing methods typically employ uniform sampling in the angle domain and non-uniform sampling in the range domain. In particular, in~\cite{10123941,10403506}, the maximum admissible range for a given angle pair is determined by a prescribed correlation threshold. However, adopting a lower correlation threshold reduces the maximum allowable sampling distance, which may result in coverage holes in the near-field region. Moreover, this maximum range is agnostic to the \ac{EBRD} boundary; consequently, complete near-field coverage is not guaranteed. 

In contrast, polar codebooks designed based on the beamdepth and the \ac{EBRD} can effectively reduce the column coherence $\mu$ while ensuring uniform $\unit[3]{dB}$ near-field coverage. In the proposed design, distance sampling starts from a minimum range, and subsequent distance samples are generated until the \ac{EBRD} limit is reached. Therefore, unlike existing methods, the proposed codebook guarantees complete near-field coverage by construction. The distance-based correlation $\mu_d = \left|\br^{\mathsf{H}}\left(\varphi_{p}, \theta_{p}, r_{p}\right) \br\left(\varphi_{p}, \theta_{p}, r_{q}\right)\right|$ demonstrates unique characteristics within the radiation field regions defined by the \ac{EBRD}. In the far-field, $\mu_d \geq 0.5$, making the samples highly correlated in the range dimension. In contrast, $\mu_d$ varies with distance relative to the \ac{EBRD} in the near-field. For near-field array response vectors sampled from outside the \ac{EBRD}, $\mu_d \geq 0.5$ persists. However, for vectors inside the \ac{EBRD} and spaced such that the 3 dB points of adjacent beams intersect, $\mu_d \leq 0.1$ is ensured, providing both $\unit[3]{dB}$ coverage and reduced column coherence. To illustrate, Fig.~\ref{fig12_codeword_correlation} shows near-field codewords for a \ac{ULA} at boresight where the adjacent beams intersect at the $0.5$ correlation level, ensuring that any \ac{LoS} user experiences a minimum beamforming gain of $\unit[3]{dB}$. Each beam’s focal point coincides with a null of the other beams, yielding $\mu_d \approx 0.08$. This value of $0.08$ can also be obtained by evaluating the near-field correlation function of a \ac{ULA}, 
$\mathcal{G}(\gamma) = \left|\frac{C^2(\gamma) + S^2(\gamma)}{\gamma^2}\right|$, at the first null location. Since $C^2(\gamma) + S^2(\gamma)$ does not reach zero exactly, the first null is approximated by its first minimum, which occurs at $\gamma_{\text{null}} \approx 1.87$, where $\mathcal{G}_{\text{null}} \approx \frac{C^2(1.87) + S^2(1.87)}{1.87^2} \approx 0.08$. The same distance sampling principle is then extended to a \ac{URA}, as described in Section~\ref{sec-6a-2}.

\begin{figure}[t!]
\centering
\includegraphics[width = 1\columnwidth]{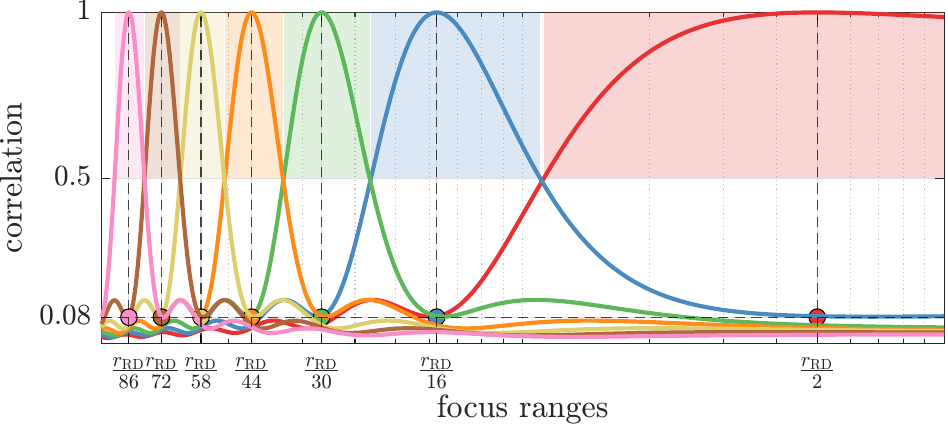}
\caption{Correlation between near-field codewords sampled based on proposed distance sampling criteria.}
\label{fig12_codeword_correlation}
\end{figure}

\subsection{Proposed Polar Codebook} \label{sec-6a}
Leveraging beamdepth and \ac{EBRD} expressions derived in Section \ref{sec-3}, we construct the proposed polar codebook $\mathbf{\Phi} \in\mathbb{C} ^ {N_1N_2 \times S}$, where $S$ is the total number of polar samples. The proposed scheme relies on sampling a grid of polar points within the \ac{EBRD} region, followed by generating a near-field codeword for each of the polar points. The algorithmic implementation of the proposed polar codebook is illustrated in Algorithm \ref{algo_pcodebook}. 
\subsubsection{Angle Sampling} \label{sec-6a-1}
For a fixed angle-dependent distance ring, the inner product between two near-field array response vectors $\left |\br^{\mathsf{H}}\left(\varphi_{p}, \theta_{p}, r_{p}\right) \br\left(\varphi_{q}, \theta_{q}, r_{q}\right)\right|$ reduces to sinc functions. To achieve orthogonal codewords, the inner product can be equated to zero, yielding the expressions for azimuth and elevation angles sampling as given in \eqref{eqn_IV_2} and \eqref{eqn_IV_3}, respectively \cite{10403506}. In line 4 of Algorithm \ref{algo_pcodebook}, we sample the elevation angle $\theta_{n_2}$, and then use it to sample the azimuth angle $\varphi_{n_1}$ in line 5. 
\subsubsection{Distance Sampling} \label{sec-6a-2}
For a given pair of azimuth and elevation angles, multiple range points are sampled. The range sampling procedure starts with the minimum range set as $2D$ in line 6. Then, in line 7, \ac{EBRD} is computed as it varies with the azimuth $\varphi$ and elevation angle $\theta$. \ac{EBRD} acts as the maximum limit for range sampling. The sampling interval between consecutive distance samples is based on the intersection of the 3 dB points. We sample the range points such that $\rf^{\mathrm{max}}$ in \eqref{eqn_IIIB_6} for the current point coincides with the $\rf^{\mathrm{min}}$ in \eqref{eqn_IIIB_7} for the next point. In line 11, we compute the upper limit of beamdepth $\tilde{\rf} ^{\mathrm{max}}$ for the current range sample $\tilde{\rf}$. Now, we want to find the next point such that its lower limit of beamdepth $\rf^{\mathrm{min}}$, coincides with $\tilde{\rf} ^{\mathrm{max}}$. To do so, we evaluate \eqref{eqn_IIIB_7} by substituting $\rf^{\mathrm{min}} = \tilde{\rf} ^{\mathrm{max}}$ to obtain the focus distance $\tilde{\rf}$. The steps from lines 10-13 are repeated until the \ac{EBRD} limit is reached. $S^{(n_1,n_2)}$ denotes the total number of distance samples for a given pair of azimuth $\varphi_{n_1}$ and elevation angle $\theta_{n_2}$.

\IncMargin{0.5em}
\begin{algorithm}[t!]\footnotesize
\caption{Polar Codebook Design for URA}\label{algo_pcodebook}
\SetKwInOut{Input}{Input}
\SetKwInOut{Output}{Output}
\DontPrintSemicolon
\SetKwFunction{CPG}{Build Polar Codebook }
\renewenvironment{algomathdisplay}
\opthanginginout 

\Input{$\lambda, N_1, N_2, D, \RD$}

\Output{$\boldsymbol{\Phi}$}

\vspace{2pt}
\hrule 
\vspace{2pt}

\SetKwProg{myproc}{Procedure}{}{end}
\myproc{\CPG}{
\For{$n_1 = -\tfrac{N_1}{2}$ \KwTo $\tfrac{N_1}{2}$}{

\For{$n_2 = -\tfrac{N_2}{2}$ \KwTo $\tfrac{N_2}{2}$}{

$ \theta_{n_2} \leftarrow u^{(n_2)}_z$ { // Select $\theta_{n_2}$ according to \eqref{eqn_IV_3}}

$\varphi_{n_1} \leftarrow u^{(n_1)}_y $ { // Select $\varphi_{n_1}$ according to \eqref{eqn_IV_2}} 

$\tilde{\rf} \leftarrow 2D$ { // Set $2D$ as minimum range}

$\EBRD \leftarrow (\varphi_{n_1},\theta_{n_2})$ {// Compute EBRD from \eqref{eqn_IIIB_8}}

$r_{s}$ $\leftarrow$ $\{\emptyset\}$, $s \leftarrow1$
 
\While{$\tilde{\rf} \leq \EBRD$}{
$r_{s} \leftarrow \tilde{\rf}$

 $\tilde{\rf} ^{\mathrm{max}} \leftarrow r_{s}$ {// Compute upper limit of beamdepth $\tilde{\rf} ^{\mathrm{max}}$ for the focus distance $r_{s}$ according to \eqref{eqn_IIIB_6}}
 
 $\tilde{\rf} \leftarrow \tilde{\rf} ^{\mathrm{max}} $ {// Compute the next distance sample $\tilde{\rf}$ by replacing $\rf^{\mathrm{min}}$ in \eqref{eqn_IIIB_7}} with $\tilde{\rf} ^{\mathrm{max}}$

 $s \leftarrow s+1$
 } 
 
$S^{(n_1,n_2)} \gets s$ { // Total distance samples for $(\varphi_{n_1},\theta_{n_2})$} 

\For{$s = 1$ \KwTo $S^{(n_1,n_2)}$ }{
$\br\leftarrow$ $\br(\varphi_{n_1}, \theta_{n_2}, r_{s}) \odot \mathbf{g}(n) $ { // Generate $\br$ from \eqref{eqn_B3} and apply weight $\mathbf{g}(n)$}
 
$\boldsymbol{\Phi}_{(n_1,n_2)}=\boldsymbol{\Phi}_{(n_1,n_2)} \cup \br$
} 
}
}
 $\boldsymbol{\Phi} = \Big[\boldsymbol{\Phi}_{(1,1)},\cdots,\boldsymbol{\Phi}_{(n_1,n_2)},\cdots,\boldsymbol{\Phi}_{(N_1,N_2)}\Big]^{\mathsf{T}}$
 
}
\end{algorithm}\DecMargin{0.5em}

\subsubsection{Build Dictionary} \label{sec-6a-3}
 For each polar coordinate $(\varphi_{n_1},\theta_{n_2}, r_s)$, we obtain the near-field array response vector $\br\in\mathbb{C} ^ {\NBS \times 1}$ in line 17, and apply weight $\mathbf{g}(n)$ as defined in \eqref{eqn_IIID_3}, to reduce axial forelobes. The choice of window function $\mathbf{f}(n)$ may vary depending upon the desired sidelobe levels and mainlobe width. Then, the $(\NBS \times S^{(n_1,n_2)})$ sub-matrix $\boldsymbol{\Phi}_{(n_1,n_2)} = \left[\br\left({\varphi}_{n_1},{\theta}_{n_2}, {r}_{1}\right), \cdots, \br\left({\varphi}_{n_1},{\theta}_{n_2}, {r}_{S^{(n_1,n_2)}}\right)\right],$ is constructed by concatenating $S^{(n_1,n_2)}$ near-field array response vectors $\br$ for $S^{(n_1,n_2)}$ distance samples. The complete polar codebook $\boldsymbol{\Phi}$ for $S$ polar samples is obtained by concatenating the sub-matrix $\boldsymbol{\Phi}_{(n_1,n_2)}$ as
\begin{equation} 
\boldsymbol{\Phi} = \Big[\boldsymbol{\Phi}_{(1,1)},\cdots,\boldsymbol{\Phi}_{(n_1,n_2)},\cdots,\boldsymbol{\Phi}_{(N_1,N_2)}\Big]^{\mathsf{T}}.
\label{eqn_IV_4}
\end{equation}

\subsection{CS-based Channel Estimation} \label{sec-6b}
Channel estimation in \ac{UM}-\ac{MIMO} systems, particularly with \ac{HBF}, poses significant challenges due to the limited number of \ac{RF} chains. Since \ac{THz} and \ac{mmWave} \ac{MIMO} channels exhibit sparsity in the frequency domain \cite{hong2021role, xing2021millimeter}, \ac{CS}-based approaches can be employed to estimate near-field channels with reduced pilot overhead \cite{9693928}. Unlike the far-field, where sparsity exists in the angular domain, the near-field channel demonstrates sparsity in the polar domain, which is also confined to the \ac{EBRD} region \cite{10988573}. To leverage this, the near-field channel $\mathbf{h}[k]$ for $\nth{k}$ subcarrier can be sparsely represented in the polar domain as $\mathbf{h}[k] = \mathbf{\Phi} \mathbf{h}^\rho[k]$, where $\mathbf{\Phi} \in \mathbb{C}^{\NBS \times S}$ denotes the near-field codebook, $\mathbf{h}^\rho[k] \in \mathbb{C}^{S \times 1}$ is the sparse near-field channel vector in the polar domain. The near-field channel can then be estimated using a \ac{CS}-based algorithm such as \ac{BF-SOMP} in \cite{10988573}, where we utilize the polar codebook proposed in Algorithm \ref{algo_pcodebook}.

\subsection{Performance Evaluation of Proposed Polar Codebook} \label{sec-6c}
We conduct Monte Carlo simulations to evaluate the performance of the proposed polar codebook in multiuser channel estimation. The simulation parameters are listed in Table \ref{table:1}. A downlink \ac{OFDM} system is considered, operating at a carrier frequency $\unit[15]{GHz}$ with a bandwidth of $\unit[100]{MHz}$ and $64$ subcarriers. The \ac{BS} simultaneously serves $\MUE=4$ single-antenna \acp{UE} using hybrid precoding. The \ac{BS} is equipped with $\NBS=540$ antenna elements and $\NRF=4$ RF chains. We assume $L=3$ paths. Channel estimation performance is quantified in terms of the \ac{NMSE}, defined as $\text{NMSE} = \mathbb{E}\!\left[\frac{\|\mathbf{H}-\hat{\mathbf{H}}_l\|_{2}^{2}}{\|\mathbf{H}\|_{2}^{2}}\right]$. The proposed scheme, referred to as \ac{BF-SOMP}, is compared against state-of-the-art polar codebooks, including P-SOMP~\cite{10123941}, EB-SOMP~\cite{10403506}, and a conventional far-field \ac{DFT} codebook.

The \ac{NMSE} performance against \ac{SNR} is illustrated in Fig. \ref{fig12_NMSE_vs_SNR}. The \ac{UE} distance from the \ac{BS} is uniformly distributed as $\mathcal{U}[2D, \EBRD]$. We assume transmission of orthogonal pilots; therefore, channel estimation for each \ac{UE} is carried out independently \cite{10273424}. The codebook size is kept the same across all schemes to ensure a fair comparison. As observed from Fig. \ref{fig12_NMSE_vs_SNR}, BF-SOMP yields $\unit[2]{dB}$ \ac{NMSE} improvement at high \ac{SNR} values compared to the other polar codebooks. The \ac{NMSE} performance with respect to the number of pilots is shown in Fig.~\ref{fig13_NMSE_vs_pilot}. The pilot length is varied from $24$ to $120$ while the \ac{SNR} is fixed at $\unit[10]{dB}$. The BF-SOMP consistently outperforms the other polar codebooks, whose \ac{NMSE} saturates with increasing pilot length, whereas the BF-SOMP continues to improve as the number of pilots increases.

\begin{table}[t]
\caption{Simulation Parameters.}
\centering
\begin{tabular}{ |c|c| }
\hline
\textbf{Parameter} & \textbf{Value} \\[0.20ex]
\hline
URA geometry & $30 \times 18$ \\[0.20ex]
\hline
Number of BS antennas ($\NBS$) & 540 \\[0.20ex]
\hline
Number of RF chains ($N_{\mathrm{RF}})$ & 4 \\[0.20ex]
\hline
Carrier frequency ($f_{c}$) & $\unit[15]{GHz}$ \\[0.20ex]
\hline
Number of subcarriers ($K$) & 64 \\[0.20ex]
\hline
Bandwidth ($B$) & $\unit[100]{MHz}$ \\[0.20ex]
\hline
Number of users ($\MUE$) & 4 \\[0.20ex]
\hline
Number of paths ($L$) & $3$ \\[0.20ex]
\hline
Users distribution in azimuth & $\left(-\tfrac{\pi}{6}, \tfrac{\pi}{6}\right)$ \\[0.20ex]
\hline
Users distribution in elevation & $\left(\tfrac{\pi}{3}, \tfrac{2\pi}{3}\right)$ \\[0.20ex]
\hline
Signal-to-noise ratio (SNR) & $\unit[10]{dB}$ \\[0.20ex]
\hline
Pilot training length (Q) & $64$ \\[0.20ex]
\hline
Number of Monte Carlo trials & $1000$ \\[0.20ex]
\hline
\end{tabular}
\label{table:1} 
\end{table}

\subsection{Computational Complexity and Column Coherence} \label{sec-6D}
The storage requirements and computational complexity of the proposed polar codebook and channel estimation algorithm are compared with FF-SOMP, P-SOMP \cite{10123941}, and EB-SOMP \cite{10403506}, as summarized in Table~\ref{table:2}.

\subsubsection{Storage Requirements}
The azimuth and elevation directional cosines are uniformly quantized as $u_y \in [-1,1]$ and $u_z \in [-1,1]$ into $N_1$ and $N_2$ samples, respectively, yielding a total of $\NBS = N_1 N_2$ angular grid points. For a planar array, however, only those samples that lie within the visible region defined by $u_y^2 + u_z^2 \leq 1$ correspond to physically observable directions. Since the visible region is a unit disk of area $\pi$ embedded within a square of area $4$, the effective number of angular samples is approximately $\frac{\pi}{4} N_1 N_2$. The angular sampling strategy is identical for all considered codebooks; hence, differences in storage requirements arise solely from their respective range sampling methods. Both the EB-SOMP and P-SOMP codebooks initialize range sampling from a maximum distance $Z_\Delta$, determined by a correlation threshold $\Delta$, and subsequently select distances as $Z_\Delta/s$, where $s = 1,2,3,\ldots$, until a predefined minimum distance threshold (set to $2D$ here) is reached. For both EB-SOMP and P-SOMP, the maximum sampling distance is chosen close to the \ac{EBRD}, and the resulting codebook size and coherence are analyzed. By construction, the proposed polar codebook restricts its range sampling exclusively to the \ac{EBRD} region. Its total size is given by $S_{\mathrm{BF}} = \sum_{(n_1,n_2)\in\mathcal{V}} S^{(n_1,n_2)}$, where $\mathcal{V}$ denotes the visible region. For elongated \acp{URA}, $S^{(n_1,n_2)} = 6$ for directions close to boresight and gradually decreases to $S^{(n_1,n_2)} = 1$ toward endfire directions. To illustrate the resulting storage requirements, we present a numerical example considering four \ac{URA} configurations with $N_2 = 64$ fixed and $N_1 \in \{8,16,24,32\}$, corresponding to total antenna counts of $\NBS \in \{512,1024,1536,2048\}$. Fig.~\ref{fig14_codebook_size} shows the resulting codebook sizes as a function of the total number of antenna elements. The proposed polar codebook consistently achieves the lowest storage requirement across all configurations. In particular, for $\NBS = 2048$, the codebook size of the proposed BF-SOMP scheme is approximately half that of the P-SOMP codebook.

\begin{figure}[t!]
\centering
\includegraphics[width = 1\columnwidth]{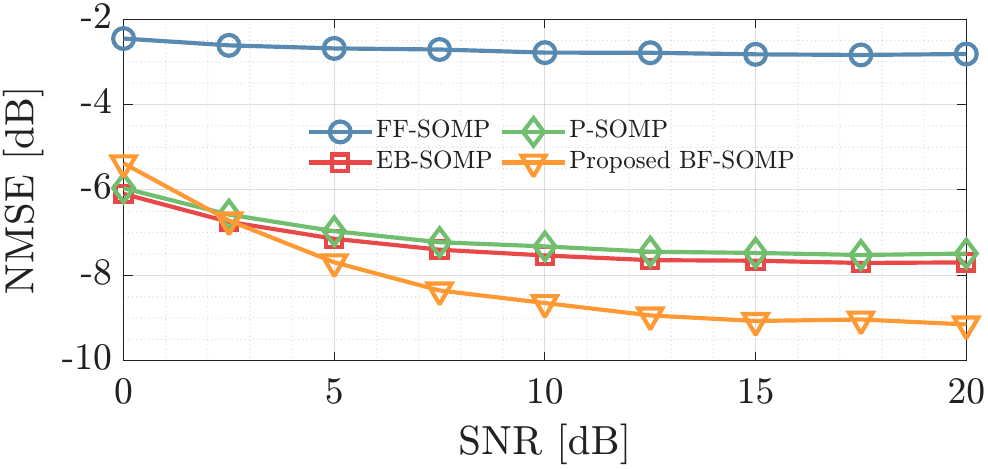}
\caption{NMSE performance with respect to SNR.}
\label{fig12_NMSE_vs_SNR}
\end{figure}

\begin{figure}[t!]
\centering
\includegraphics[width = 1\columnwidth]{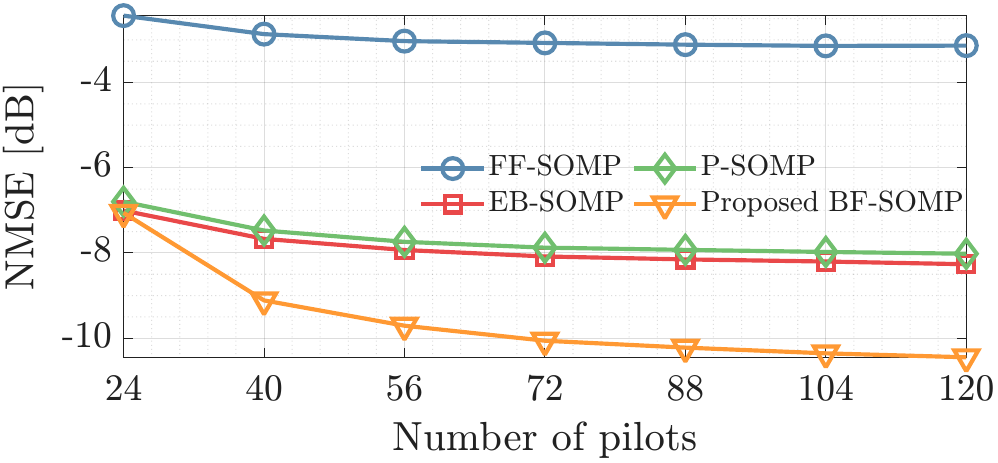}
\caption{NMSE performance with respect to pilot length.}
\label{fig13_NMSE_vs_pilot}
\end{figure}

\subsubsection{Computational Complexity}
We assume that channel estimation is carried out using the SOMP algorithm. The correlation step primarily dominates its computational complexity and can
be expressed as $\mathcal{O}\!\left({L}\, Q\, \NRF\, \NBS\, S_{\mathrm{BF}}\, K\right)$, where the received pilot matrix of dimension $Q \NRF \times K$ is correlated
with a codebook of size $\NBS \times S_{\mathrm{BF}}$ over $L$ propagation paths. Consequently, the differences in computational complexity among the considered schemes arise directly from the size of their respective codebooks, as summarized in Table~\ref{table:2}.

\subsubsection{Column Coherence}
The column coherence defined in~\eqref{eqn_IV_1} is evaluated for all schemes using the same simulation parameters as in the codebook size comparison. The proposed BF-SOMP exhibits a column coherence of approximately $0.26$, compared to about $0.45$ for P-SOMP and $0.4$ for EB-SOMP. For the proposed BF-SOMP, the range samples are selected based on the beamdepth such that their mutual correlation remains below $0.1$, as discussed at the beginning of this section. The coherence value of approximately $0.26$ originates from the angular domain, since the angular samples are not perfectly orthogonal. Achieving near-zero correlation in the angular domain would require the range samples associated with the same angle to be selected from distinct distance rings, as adopted in~\cite{10403506}. While this approach improves angle domain orthogonality, it introduces stronger correlations among the range samples. In contrast, the proposed method samples the range dimension based on the beamdepth criterion, which simultaneously reduces the overall coherence, ensures complete near-field coverage, and lowers the storage requirements.

\begin{table}[t!]
\caption{Codebook size and computational complexity.}
\begin{center}
\begin{tabular}{ |c|c|c| }
\hline
\textbf{Algorithm} & \textbf{Codebook Size} & \textbf{Computational Complexity} $\mathcal{O}(\cdot)$ \\
\hline
FF-SOMP & ${\NBS \times \NBS}$ & $\mathcal{O}\left({L} Q{\NRF}{\NBS}K\right)$\\
\hline
BF-SOMP & ${\NBS \times S_{\mathrm{BF}}}$ & $\mathcal{O}({L} Q{\NRF}S_{\mathrm{BF}}K)$\\
\hline
P-SOMP \cite{10123941} & ${\NBS \times S_{\mathrm{P}}}$ & $\mathcal{O}({L} Q{\NRF} S_{\mathrm{P}} K)$\\
\hline
EB-SOMP \cite{10403506} & ${\NBS \times S_{\mathrm{EB}}}$ & $\mathcal{O}({L} Q{\NRF}S_{\mathrm{EB}}K)$\\
\hline
\end{tabular}
\end{center}
\label{table:2} 
\end{table}

\begin{figure}[t!]
\centering
\includegraphics[width = 1\linewidth]{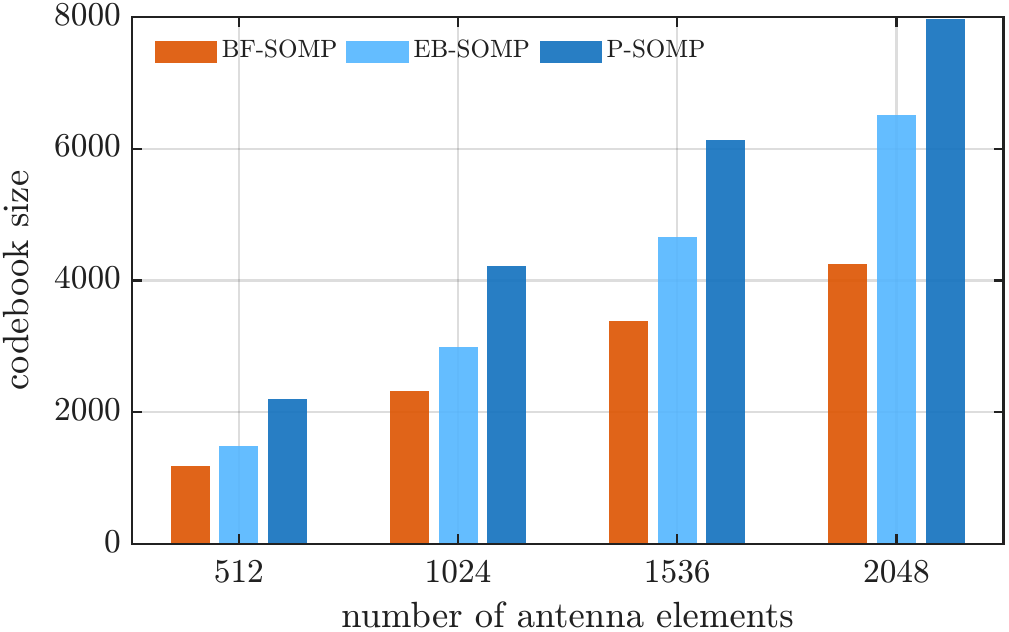}
\setlength{\belowcaptionskip}{-15pt}
\setlength{\abovecaptionskip}{0pt}
\caption{Codebook size as a function of number of antenna elements.}
\label{fig14_codebook_size}
\end{figure}

\section{Conclusion} \label{sec-7}
In this paper, we have derived the angle-dependent beamdepth and the \ac{EBRD} for a generalized \ac{URA}. The results show that, for a fixed number of elements, the array geometry has a pronounced impact on the beamdepth, whereas under a fixed aperture constraint, this effect becomes negligible. The near-field beam pattern analysis further reveals that the lateral resolution, determined by the beamwidth, surpasses the axial resolution, governed by the beamdepth. Moreover, amplitude weighting introduces a tradeoff between axial and lateral sidelobe levels. We have also derived an \ac{EDoF} expression for \ac{LoS} \ac{MIMO} systems and compared it across different array geometries, showing that the \ac{ULA} achieves a higher \ac{EDoF} than a square \ac{URA}. Finally, building on the beamdepth and \ac{EBRD} analysis, we have proposed a polar codebook that enhances channel estimation performance while reducing computational complexity.
 
\bibliographystyle{IEEEtran}
\bibliography{Bibliography/IEEEabrv,Bibliography/my2bib}
\appendices

\section{ Proof of Theorem \ref{theorem1}} 
\label{Appendix_A}
We define $\alpha_{\mathrm{\scalebox{0.5}{3dB}}} \stackrel{\Delta}{=}\left\{ \left(\gamma_1\gamma_2\right) \mid \GR\left(\gamma_1, \gamma_2 \right) = 0.5 \right\}$. Thus, $\alpha_\mathrm{\scalebox{0.5}{3dB}} = \frac{N_1 N_2 d^2}{2\lambda} \sqrt{ \beta_1 \beta_2} z_\mathrm{eff}$. Substituting $d^2 = D^2/(N_1^2 + N_2^2)$, $\alpha_\mathrm{\scalebox{0.5}{3dB}}$ can be obtained as
\begin{equation}
\begin {aligned}
\alpha_\mathrm{\scalebox{0.5}{3dB}}& = \frac{\RD N_1 N_2}{4(N_1^2 + N_2^2)}\sqrt{\beta_1 \beta_2} z_\mathrm{eff}\\
 &= \frac{\RD }{4} \left(\frac{\eta}{\eta^2+1}\right)\sqrt{ \beta_1 \beta_2} z_\mathrm{eff},
\end {aligned}
\label{eqn_IIIB_5}
\end{equation}
 where $\eta = N_1/N_2$. We can then solve for $z$ in (\ref{eqn_IIIB_5}) to get $ z = \frac{\rf \RD \sqrt{ \beta_1 \beta_2} \eta }{\RD \eta \sqrt{ \beta_1 \beta_2} \pm 4 \rf \alpha_\mathrm{\scalebox{0.5}{3dB}}(\eta^2+1)}$. Hence, 
\begin{subequations}
\begin{equation}
\rf^\mathrm{max} = \frac{\rf \RD \sqrt{ \beta_1 \beta_2} \eta }{\RD \eta \sqrt{ \beta_1 \beta_2} - 4 \rf \alpha_\mathrm{\scalebox{0.5}{3dB}}(\eta^2+1)}, 
\label{eqn_IIIB_6}
\end{equation}
\begin{equation}
 \rf^\mathrm{min} = \frac{\rf \RD \sqrt{ \beta_1 \beta_2} \eta }{\RD \eta \sqrt{ \beta_1 \beta_2} + 4 \rf \alpha_\mathrm{\scalebox{0.5}{3dB}}(\eta^2+1)}.
\label{eqn_IIIB_7}
\end{equation}
\end{subequations} 
The distance window between $\rf^\mathrm{max}$ and $\rf^\mathrm{min}$ is the interval where $\GR$ is less than or equal to $\unit[3]{dB}$. Therefore, $\BD=\rf^\mathrm{max}-\rf^\mathrm{min}$ is given by (\ref{eqn_IIIB_4}), which completes the proof. 

\section{ Proof of Theorem \ref{theorem2}} 
\label{Appendix_B}
In \eqref{eqn_IIIB_4}, the maximum value of $\BD$ is obtained when the factor in the denominator {\small $ \left[\eta \RD \sin \theta \sqrt{1-\sin^2 \theta \sin^2 \varphi}\right]^2 -\left[4\alpha_\mathrm{\scalebox{0.5}{3dB}} \rf (\eta^2+1)\right]^2 = 0$}. Thus, the farthest angle-dependent axial distance \(\rf\), within which finite-depth beamforming can be achieved, is given by 
\(\rf < \frac{\eta \RD}{4 \alpha_\mathrm{\scalebox{0.5}{3dB}} (1+\eta^2)} \sin \theta \sqrt{1-\sin^2 \theta \sin^2 \varphi}.\) Otherwise, for distances exceeding this limit, the beamdepth approaches infinity (\(\BD \rightarrow \infty\)).

\section{ Proof of Theorem \ref{Theorem-4}} 
\label{Appendix_C}

The near-field array response vector for the $\nth{n}$ antenna element of a \ac{ULA} is $\mathbf{b}_n (\varphi,r) \approx \tfrac{1}{\sqrt{N}}e^{-j\nu\{nd\sin(\varphi)- \frac{1}{2r}n^2d^2\cos^2(\varphi)\}}$. The tapered axial beam pattern can be obtained by $\mathcal{G}(\varphi, r) = \left| \left( \mathbf{g} \odot \mathbf{b}(\varphi,\rf) \right)^\mathsf{H} \mathbf{b}(\varphi,r) \right|^2$, yielding 
\begin{equation} 
\mathcal{G}(\varphi,r) = \left| \frac{1}{\NBS} \int_{-\NBS/2}^{\NBS/2} {\mathbf{g}}(n) \, e^{ -j n^2 \beta_1 } \, dn \right|^2, 
\label{eqn-B8-19}
\end{equation}
where $\beta_1 = \nu d^2 \cos^2(\varphi) \, r_\mathrm{eff}$. In the far-field, the lateral pattern of $\mathbf{f}(n)$ is obtained by computing its Fourier transform as 
\begin{equation}
\mathcal{G}(\theta) = \left| \int_{0}^{\NBS} \mathbf{f}(n) \, e^{ -j n \beta_2} \, dn \right|^2,
\label{eqn-B9-20}
\end{equation}
where $\beta_2 = \nu d \sin (\theta)$. The proposed synthesis method reformulates \eqref{eqn-B8-19} to resemble \eqref{eqn-B9-20}, allowing the desired weighting function $\mathbf{f}(n)$ to be specified and then mapped back to ${\mathbf{g}}(n)$ that can suppress axial forelobes. Introducing the transformation $t = n^2$, the integral in \eqref{eqn-B8-19} can equivalently be rewritten as
\begin{equation}
\mathcal{G}(\theta,r) = \left| \frac{1}{\NBS} \int_{t \in \mathcal{T}} \frac{ \bar{\mathbf{w}}(+\sqrt{t}) + \bar{\mathbf{w}}(-\sqrt{t}) }{2\sqrt{t}} \, e^{-j t \beta_1} \, dt \right|^2,
\label{eqn-B10-21}
\end{equation}
where $\mathcal{T} = [0, (\NBS/2)^2]$. Since \eqref{eqn-B10-21} has a form similar to \eqref{eqn-B9-20}, the transformation ${\mathbf{g}}(n) = |n| \, {\mathbf{f}}(n^2)$ can be employed to derive the window function to reduce the axial forelobes.
\section{ Proof of Theorem \ref{theorem_3}} 
\label{Appendix_D}
We rewrite the summation over $N_1$ in (\ref{eqn_V_3}) by completing the square as
\begin{equation}
\left|\frac{1}{N_1} \sum_{-\tfrac{N_1}{2}}^{\tfrac{N_1}{2}} e^{-j \pi\left(A_{1} n_1 -A_{2}\right)^{2}}\right|^2=\left|F\left(A_{1}, A_{2}\right)\right|^2,
 \label{eqn20}
\end{equation}
where $A_{1}=\sqrt{\frac{d\left(1-u_y^{2}\right)}{2 \rf}}$ and $A_{2}=\frac{1}{A_{1}}\left(\frac{u_y-u^{(n_1)}_y}{2}\right)$.
 Then, the summation in $F\left(A_{1}, A_{2}\right)$ can be approximated as
 \begin{equation}
 \begin{aligned}
F\left(A_{1}, A_{2}\right) & \stackrel{\left(b_{1}\right)}{\approx} \frac{1}{N_1} \int_{-\tfrac{N_1}{2}}^{\tfrac{N_1}{2}} e^{-j \pi\left(A_{1} n_1-A_{2}\right)^{2}} \mathrm{~d}n_1 \\
& \stackrel{\left(b_{2}\right)}{=} \frac{1}{ \sqrt{2} A_{1} N_1} \int_{-\tfrac{1}{\sqrt{2}} A_{1} N_1-\sqrt{2} A_{2}}^{\tfrac{1}{\sqrt{2}} A_{1} N_1-\sqrt{2} A_{2}} e^{j \pi \frac{t^{2}}{2} } \mathrm{~d}t 
\label{eqn21}
 \end{aligned}
 \end{equation}
where $\left(b_{1}\right)$ is accurate when $N_1 \rightarrow \infty$, and $\left(b_{2}\right)$ is obtained by letting $A_{1} n_1-A_{2}=\frac{1}{\sqrt{2}} t$. Next, based on the Fresnel integrals, we have $F\left(A_{1}, A_{2}\right)$ as
\begin{equation}
\begin{aligned}
& =\frac{\int_{0}^{{\tfrac{1}{\sqrt{2}} A_{1} N_1-\sqrt{2} A_{2}}} e^{j \pi \frac{t^{2}}{2} } \mathrm{~d} t-\int_{0}^{{-\tfrac{1}{\sqrt{2}} A_{1} N_1\sqrt{2} A_{2}}} e^{-j \pi \frac{t^{2}}{2} } \mathrm{~d} t}{\sqrt{2} A_{1} N_1} \\
& =\frac{C\left(\gamma_{1}+\gamma_{2}\right)-C\left(\gamma_{1}-\gamma_{2}\right)+j\left(S\left(\gamma_{1}+\gamma_{2}\right)-S\left(\gamma_{1}-\gamma_{2}\right)\right)}{2 \gamma_{2}}
\end{aligned} 
\end{equation}
where $\gamma_{1}=(u^{(n_1)}_y-u_y) \sqrt{\frac{\rf}{d\left(1-u^{2}_y\right)}}, \gamma_{2}=\frac{N_1}{2} \sqrt{\frac{d\left(1-u_{y}^{2}\right)}{\rf}}$.
The factor involving summation over $N_2$ can be evaluated analogously, thereby completing the proof.
\end{document}